\documentclass{article}
\usepackage{authblk}
\usepackage{amsmath,amssymb,amsthm,amsfonts,bm,fancybox,fullpage}
\usepackage{multirow}
\usepackage[dvipdfmx]{graphicx}
\usepackage{color}
\usepackage{hyperref}
\usepackage{booktabs}
\usepackage{cite}
\usepackage[
	n,
	operators,
	advantage,
	sets,
	adversary,
	landau,
	probability,
	notions,	
	logic,
	ff,
	mm,
	primitives,
	events,
	complexity,
	asymptotics,
	keys
	]{cryptocode}

\theoremstyle{plain}
\newtheorem{theorem}{Theorem}[section]
\newtheorem{lemma}[theorem]{Lemma}
\newtheorem{corollary}[theorem]{Corollary}
\newtheorem{proposition}[theorem]{Proposition}

\newtheorem{definition}[theorem]{Definition}

\theoremstyle{definition}

\newcommand{\ra}{\rightarrow}
\newcommand{\la}{\leftarrow}

\newlength{\ralen}
\newcommand{\blk}{\raisebox{\ralen}{\framebox(11,13){$\clubsuit$}}}
\newcommand{\red}{\raisebox{\ralen}{\framebox(11,13){$\heartsuit$}}}
\newcommand{\back}{\raisebox{\ralen}{\framebox(11,13){\bf\large ?}}}
\newcommand{\henc}{E^{\heartsuit}}

\newcommand{\card}{\mathsf{card}}
\newcommand{\rand}{\mathsf{rand}}
\newcommand{\comm}{\mathsf{comm}}

\begin{document}

\title{Single-Shuffle Full-Open Card-Based Secure Computation Protocols for Any Function}
\author[1*]{Reo Eriguchi}
\author[2,3,1]{Kazumasa Shinagawa}
\affil[1]{National Institute of Advanced Industrial Science and Technology, Japan}
\affil[2]{University of Tsukuba, Japan}
\affil[3]{Kyushu University, Japan}
\affil[*]{Corresponding author: eriguchi-reo@aist.go.jp}
\date{}
\maketitle

\begin{abstract}
A card-based secure computation protocol is a method for $n$ parties to compute a function $f$ on their private inputs $(x_1,\ldots,x_n)$ using physical playing cards, in such a way that the suits of revealed cards leak no information beyond the value of $f(x_1,\ldots,x_n)$.
A \textit{single-shuffle full-open} protocol is a minimal model of card-based secure computation in which, after the parties place face-down cards representing their inputs, a single shuffle operation is performed and then all cards are opened to derive the output.
Despite the simplicity of this model, the class of functions known to admit single-shuffle full-open protocols has been limited to a few small examples.
In this work, we prove for the first time that every function admits a single-shuffle full-open protocol.
We present two constructions that offer a trade-off between the number of cards and the complexity of the shuffle operation.
These feasibility results are derived from a novel connection between single-shuffle full-open protocols and a cryptographic primitive known as \textit{Private Simultaneous Messages} protocols, which has rarely been studied in the context of card-based cryptography.
We also present variants of single-shuffle protocols in which only a subset of cards are revealed.
These protocols reduce the complexity of the shuffle operation compared to existing protocols in the same setting. 
\end{abstract}


\section{Introduction}
\label{section:introduction}

A \textit{secure computation} protocol \cite{Yao82} is a fundamental primitive in cryptography and has been studied extensively in the literature (e.g, \cite{Lin20} and references therein).
This primitive enables $n$ parties, each holding a private input $x_i$, to evaluate a function $f$ on $(x_1,\ldots,x_n)$ without revealing any information beyond the value of $f(x_1,\ldots,x_n)$.
\textit{Card-based secure computation} \cite{BoerEC1989,KilianC1994,MizukiIJIS2014} studies implementations of secure computation protocols using physical playing cards.
In this work, we focus on binary cards with suits $\heartsuit$ or $\clubsuit$, which is the most standard setting in the literature.
This visual implementation not just serves as an accessible tool that helps non‑experts understand the mechanisms of secure computation protocols, but also has found interesting theoretical connections to discrete mathematics and other cryptographic primitives \cite{ShinagawaFUN2024,ShinagawaFCT2025,SN25}.
Many works have proposed different card-based secure computation protocols for both general and specific functions, e.g., \cite{NishidaTAMC2015,KastnerAC2017,ShinagawaDAM2021,TakahashiAPKC2024}.

Most of the existing card-based secure computation protocols proceed as follows.
First, the parties place face-down cards encoding their inputs\footnote{In the standard input model, at most two cards in the sequence representing $x_i$ depend on $x_i$, while the remaining cards are independent of $x_i$. In contrast, we consider a more general setting where any number of cards in the sequence representing $x_i$ may depend on $x_i$.}.
They then apply a random shuffle, which is unknown to any party, to the sequence of cards and reveal the suits of certain cards at positions specified by the protocol.
After possibly multiple rounds of shuffle and opening operations, the parties derive an output $f(x_1,\ldots,x_n)$ from the suits revealed so far.

Note that the security of such protocols relies on two conditions: 
the permutation chosen by every shuffle must be hidden from all parties,
and the suits of cards that are not designated to be opened must remain secret.
In other words, the parties need to be careful whenever performing such operations, as improper operations can immediately compromise the privacy of their inputs.
It is therefore important to minimize the number of shuffle operations and simplify opening operations, when designing card-based protocols.

In this work, we consider the simplest setting for card-based secure computation protocols, referred to as \textit{single-shuffle full-open protocols} \cite{ShinagawaEPRINT2022,SN25}.
Specifically, in such protocols, a single shuffle operation is performed after the parties place face-down cards corresponding their inputs, and then all cards are turned face up. 
The output is subsequently computed from the fully revealed sequence of suits.
This model assumes the minimum possible number of shuffle operations and the simplest form of opening, namely, turning over all cards.
Thus, single-shuffle full-open protocols are important in that they minimize the complexity of both shuffling and opening operations.

Some feasibility results have been shown in this setting.
The first card-based secure computation protocol, known as ``Five-Card Trick'' \cite{BoerEC1989}, is a single-shuffle full-open protocol that computes the two-input AND function. 
It is also known that several specific functions such as three-input equality, majority, and XOR functions can be computed by single-shuffle full-open protocols \cite{HeatherFAOC2014,ShinagawaICISC2018,ToyodaINDOCRYPT2021,KuzumaAPKC2022,ShinagawaARXIV2025}.
However, prior to this work, the fundamental question of whether every function $f:\bin^n\ra\bin$ admits a single-shuffle full-open protocol remained open.
Surprisingly, even the three-input AND function is not known to have a single-shuffle full-open protocol.
Note that if a more advanced operation --- namely, opening only a \textit{subset} of cards ---  is allowed, then every function can be computed with a single shuffle \cite{ShinagawaDAM2021,TozawaUCNC2023,OnoICISC2023}.
However, it is not straightforward to extend such protocols to the full-open setting since revealing even a single additional card may compromise privacy.

\subsection{Our Results}

In this work, we prove that every function $f:\bin^n\ra\bin$ admits a single-shuffle full-open protocol.
Our constructions are based on a novel connection between single-shuffle full-open protocols and a cryptographic primitive known as \textit{Private Simultaneous Messages} protocols, which has rarely been studied in the context of card-based cryptography.
We also present variants of single-shuffle protocols in which only a subset of cards are opened.
These protocols reduce the complexity of the shuffle operation compared to the previous single-shuffle protocols \cite{ShinagawaDAM2021,TozawaUCNC2023,OnoICISC2023}.
We show a comparison of single-shuffle protocols in Table~\ref{table:1}.
Below, we elaborate on our results.

\begin{table}[t]
    \caption{Comparison of single-shuffle card-based secure computation protocols}
\begin{center}
    \begin{tabular}{ccccc}
    Protocol&Function&Opening&Shuffle&\# cards\\
    \toprule
    \cite{BoerEC1989}&2-input AND&full&$\text{RC}\times 1$&$5$ \\
    \cite{ShinagawaICISC2018}&2-input XOR &full&$\text{RC}\times 1$&$4$ \\
    \cite{ShinagawaARXIV2025}&3-input XOR &full&$\text{RC}\times 1$&$8$ \\
    \cite{HeatherFAOC2014,ShinagawaICISC2018}&3-input equality &full&$\text{RC}\times 1$&$6$ \\
    \cite{ToyodaINDOCRYPT2021}&3-input majority &full&$\text{RC}\times 1$ \& $\text{RBC}\times1$&$6$ \\
    \cite{KuzumaAPKC2022}&$n$-input XOR&full&$\text{RBC}\times (n-1)$&$2n$ \\
    \cite{KuzumaAPKC2022}&$n$-input XOR&full&$\text{PSc}\times 2$&$O(n\log n)$ \\
    Ours (Corollary~\ref{corollary:SSFO-all})&any&full&$\text{PSh}\times1$ \& $\text{CS}\times1$&$2^{O(n^32^{n/2})}$ \\
    Ours (Corollary~\ref{corollary:SSFO-all-2})&any&full&$\text{PSh}\times O(n2^n)$ \& $\text{PSc}\times1$&$O(n^22^n)$ \\
    Ours (Corollary~\ref{corollary:SSFO-Sym})&symmetric&full&$\text{PSh}\times n$&$O(n^32^{2n})$ \\
    Ours (Corollary~\ref{corollary:SSFO-AND})&$n$-input AND&full&$\text{PSh}\times n$&$O(n^2)$ \\
    \midrule
    \cite{MizukiAC2012}&2-input AND&adaptive&$\text{RC}\times 1$ \& $\text{RBC}\times1$&$4$ \\
    \cite{KuzumaAPKC2022}&$n$-input AND&adaptive&$\text{RBC}\times (n-1)$&$4n-2$ \\
    \cite{KuzumaAPKC2022}&$n$-input AND&adaptive&$\text{PSc}\times 2$& $O(n\log n)$ \\
    \cite{ShinagawaDAM2021,TozawaUCNC2023,OnoICISC2023}&any&adaptive&$\text{PSc}\times O(C_f+n)$&$O(C_f+n)$\\
    Ours (Corollary~\ref{corollary:SSFO-all})&any&static&$\text{RC}\times 1$&$2^{O(n^32^{n/2})}$\\
    \bottomrule
    \end{tabular}
\end{center}
\fontsize{8pt}{8pt}\selectfont
$n$ refers to the number of inputs.
``CS'' refers to a complete shuffle, 
``PSc'' refers to a pile-scramble shuffle, 
``RC'' refers to a random cut, 
``PSh'' refers to a pile-shifting shuffle, and
``RBC'' refers to a random bisection cut, which is a pile-shifting shuffle performed on two piles of cards. 
In each protocol, these shuffles are composed into a single shuffle. 
$C_f$ refers to the minimum number of gates in any circuit computing $f$.
Note that the protocols except for \cite{ShinagawaARXIV2025} and ours follow the standard input model.
\fontsize{10pt}{10pt}\selectfont
    \label{table:1}
\end{table}

\smallskip
\noindent\textbf{Single-Shuffle Full-Open Protocols for Any Function.}
We present two constructions of single-shuffle full-open protocols that offer a trade-off between the number of cards and the complexity of the shuffle operation.
To describe our constructions, we introduce basic shuffle operations.
\begin{itemize}
    \item A \textit{complete shuffle} applies a uniformly random permutation to a sequence of cards.
    \item A \textit{pile-scramble shuffle} divides a sequence of cards into multiple piles and applies a complete shuffle to these piles.
    \item A \textit{random cut} applies a uniformly random shift to a sequence of cards.
    \item A \textit{pile-shifting shuffle} divides a sequence of cards into piles and applies a random cut to the piles.
\end{itemize}
All our protocols use a composition of the above four shuffle operations\footnote{We
treat any composition of shuffles as a single shuffle.}.
We consider that a complete shuffle (resp.\ a random cut) is preferable to a pile-scramble (resp.\ a pile-shifting) shuffle since the latter requires additional items such as envelopes or sleeves to make piles.

For any function $f:\bin^n\ra\bin$, our first construction requires $2^{O(n^32^{n/2})}$ cards and uses a composition of one pile-shifting shuffle and one complete shuffle.
Our second construction reduces the number of cards to $O(n^22^n)$, but at the cost of a more complex shuffle: a composition of $O(n2^n)$ pile-shifting shuffles and one pile-scramble shuffle.

\smallskip
\noindent\textbf{Techniques: Implication from Private Simultaneous Messages.}
Our constructions are based on a novel connection between single-shuffle full-open protocols and a cryptographic primitive known as \textit{Private Simultaneous Messages} (PSM) protocols.
A PSM protocol was originally proposed by Feige et al. \cite{FKN94} as a non-interactive model of secure computation, in a setting unrelated to card-based cryptography.
In this model, each party computes a message $m_i^{x_i}(\rho)$ based on their private input $x_i$ and shared randomness $\rho\in\bin^r$, and then an external party can compute $f(x_1,\ldots,x_n)$ from the collection of messages $(m_i^{x_i}(\rho))_{i=1,\ldots,n}$.
The security requirement is that the messages reveal no information beyond the output of $f$ without the knowledge of the shared randomness $\rho$.

We present two transformations from PSM protocols to single-shuffle full-open protocols, corresponding to the above two constructions.
In retrospect, both transformations are quite simple.
In the first construction, we enumerate all possible collections of messages $M(\rho)=(m_i^{x_i}(\rho))_{i=1,\ldots,n}$ for every possible value of randomness $\rho\in\bin^r$.
Each such $M(\rho)$ is encoded into a pile of cards, resulting in $2^r$ piles in total.
A pile-shifting shuffle is applied to select a pile $M(\rho)$ corresponding to a uniformly random $\rho$, and then a complete shuffle is applied to the remaining piles to ``erase'' the other messages $M(\rho')$ for $\rho'\neq \rho$.
Finally, all cards are revealed and $f(x_1,\ldots,x_n)$ is obtained by applying the decoding algorithm of the PSM protocol to the selected pile $M(\rho)$.
Since every function admits a PSM protocol, this transformation leads to a single-shuffle full-open protocol for any function $f$.
However, the number of piles is exponential in the randomness complexity $r$ of the PSM protocol.
Since the currently best known PSM protocol has $r=O(n^32^{n/2})$ \cite{BKN18}, a doubly exponential number of cards $2^{\Omega(2^{n/2})}$ are required.

We propose an alternative transformation applicable to a subclass of PSM protocols that satisfy a certain algebraic property, which we refer to as \textit{additive PSM}.
This transformation reduces the number of cards compared to our first transformation, which applies to arbitrary PSM protocols.
Since several functions admit efficient additive PSM protocols \cite{Ish13,SESN23}, it derives single-shuffle full-open protocols with fewer cards for those functions including the AND function and any symmetric function\footnote{We say that a function is symmetric if its output is independent of the order of inputs.}.
In particular, we obtain a single-shuffle full-open protocol for the AND function using $O(n^2)$ cards.
We can extend this protocol to general functions, by leveraging the fact that any function can be computed by evaluating at most $2^n$ AND functions.
This yields our second construction, which uses $O(n^22^n)$ cards.
See Section~\ref{sec:overview} for more details.

Notably, a converse implication was shown in \cite{SN25}: Any single-shuffle full-open protocol for a function $f$ can be transformed into a PSM protocol for $f$.
Combining the result of \cite{SN25} with ours, we can relate the minimum number of cards required by any single-shuffle full-open protocol for $f$, to the minimum communication complexity and the minimum randomness complexity of any PSM protocol for $f$.
See Section~\ref{section:relation} for details.

\smallskip
\noindent\textbf{Variants: Single-Shuffle Protocols with Partial Opening.}
We show variants of our first construction that reduce the complexity of the shuffle operation by relaxing the requirement that all cards are revealed.
When a subset of cards are opened, we distinguish two types of opening: (1)~\textit{static opening}, where a fixed subset of cards are revealed; (2)~\textit{adaptive opening}, where the next card to reveal is determined by the suits of the cards that have already been revealed.
We consider that static opening is preferable to the adaptive one in terms of simplicity.

We show a single-shuffle static-open protocol for any function using a random cut only.
Technically, our full-open protocol requires a composition of a pile-shifting shuffle and a complete shuffle.
The complete shuffle is required to ``erase'' information encoded in a specific subset of cards.
If we modify the protocol in such a way that the cards except those in that subset are revealed, we obtain a single-shuffle static-open protocol using only a pile-shifting shuffle.
Note that the pile-shifting shuffle is used to reveal one pile selected uniformly at random from multiple piles.
We propose a novel protocol to realize such a functionality using a random cut, which may be of independent interest.
As a result, we obtain a single-shuffle static-open protocol for any function $f$ using a random cut.
This variant requires additional cards but the total number of cards remains asymptotically the same as our full-open protocol, i.e., $2^{O(n^32^{n/2})}$.

Our protocol reduces the complexity of the shuffle operation: the protocols in \cite{ShinagawaDAM2021,TozawaUCNC2023,OnoICISC2023} require a composition of $O(C_f+n)$ pile-scramble shuffles, where $C_f$ is the minimum number of gates in any circuit computing $f: \bin^n \rightarrow \bin$, whereas ours uses a random cut only.
Our protocol simplifies the opening operation as well: the protocols in \cite{ShinagawaDAM2021,TozawaUCNC2023,OnoICISC2023} require adaptive opening whereas ours needs static opening.
As a trade-off, the number of cards required by our protocol is larger than 
those of the previous protocols \cite{ShinagawaDAM2021,TozawaUCNC2023,OnoICISC2023} with $O(C_f+n)$ cards, which are upper bounded by $2^{O(n)}$ in the worst case.

\subsection{Related Works}

We provide a brief summary of existing works on card-based secure computation protocols.

\smallskip
\noindent\textbf{Single-Shuffle Protocols.}
Prior to this work, the class of functions known to admit single-shuffle full-open protocols were limited to a few small examples: two-input AND \cite{BoerEC1989}, two-input and three-input XOR \cite{ShinagawaICISC2018,ShinagawaARXIV2025}, three-input equality \cite{HeatherFAOC2014,ShinagawaICISC2018}, and three-input majority \cite{ToyodaINDOCRYPT2021}.
We are the first to prove that every function can be computed by a single-shuffle full-open protocol.
The protocols in \cite{ShinagawaDAM2021,TozawaUCNC2023,OnoICISC2023} can compute any function with a single shuffle but they require a more advanced operation when opening cards, namely adaptive opening.
Furthermore, the complexity of the shuffle operation is sub-optimal: 
the protocols in \cite{ShinagawaDAM2021,TozawaUCNC2023,OnoICISC2023} require a composition of many pile-scramble shuffles proportional to the number of gates plus the number of input wires in a circuit computing the function.
In contrast, our single-shuffle adaptive-open protocol uses only one random cut.

\smallskip
\noindent\textbf{Protocols Using More Than One Shuffles.}
There are card-based secure computation protocols that apply multiple rounds of shuffling and partial opening of cards.
In this model, Nishida, Mizuki, and Sone~\cite{NishidaTAMC2015} showed that any function $f:\bin^n\ra\bin$ can be computed using $O(2^n)$ rounds of shuffle and reveal operations.
The protocol achieves a constant number of cards (i.e., six) in addition to those representing inputs.
For symmetric functions, tailor-made constructions by \cite{NishidaTAMC2015,ShikataICTAC2022,TakahashiAPKC2024} reduce the number of operations and the number of cards. 
Several works have also focused on minimizing the number of cards and shuffle operations for small functions (e.g., \cite{MiyaharaTCS2020,RuangwisesTAMC2020,RuangwisesSecITC2021,OK24}).

\smallskip
\noindent\textbf{Applications.}
Shinagawa and Nuida \cite{SN25} showed that a certain type of card-based secure computation protocols can be transformed into PSM protocols for the same function.
In particular, a single-shuffle full-open protocol using $\ell$ cards implies a PSM protocol with communication complexity $\ell n$ and randomness complexity $\ell\log \ell+\ell n$.
Prior to our work, this was the only known result demonstrating a connection between card-based protocols and PSM protocols.
Our work complements their result by showing the converse implication from PSM protocols to single-shuffle full-open protocols.
Card-based protocols have also found applications beyond secure computation such as in zero-knowledge proofs \cite{NiemiFundamInf1999,GradwohlToCS2009,MiyaharaProvsec2021} and in simulating players of card games \cite{ShinagawaToCS2025}. 

\section{Overview of Our Techniques}\label{sec:overview}

We here provide an overview of our techniques.
More detailed descriptions and proofs are given in the following sections.

\subsection{Single-Shuffle Full-Open Protocols from PSM}

We denote $[n]=\{1,\ldots,n\}$ for $n\in\mathbb{N}$.
In this paper, we use binary cards whose front sides have a suit $\heartsuit$ or $\clubsuit$. 
A single-shuffle full-open protocol for a function $f:\bin^n\ra\bin$ specifies $2n$ sequences of suits $(L_i^b)_{i\in[n],b\in\bin}$.
Each party with input $x_i\in\bin$ puts face-down cards whose suits correspond to $L_i^{x_i}$ on the table.
Then, a random permutation sampled from a probability distribution $\mathcal{D}$ is applied to the face-down cards, where the permutation is unknown to any party.
Finally, parties turn over all cards and obtain an output $y=f(x_1,\ldots,x_n)$ from the revealed suits.
The privacy requirement is that the sequence of the revealed suits leaks nothing about $(x_i)_{i\in[n]}$ beyond what follows from $y$.
While we can extend the notion into the setting where only a subset of cards are turned over, we focus on protocols where all cards are opened in this section.

As mentioned earlier, our first construction is based on a cryptographic primitive known as PSM protocols \cite{FKN94,IK97}.
This primitive enables $n$ parties each holding an input $x_i$ to reveal $f(x_1,\ldots,x_{n})$ to an external party called a referee without revealing any extra information.
Specifically, it consists of three algorithms: 
A randomness generation algorithm $\mathsf{Gen}$ takes no input and outputs a random string $\rho$, which is shared among all parties except for the referee;
An encoding algorithm $\mathsf{Enc}$ takes a party's index $i\in[n]$, a party's input $x_i$, and shared randomness $\rho$ as inputs, and outputs a message $\mu_i$, which is sent to the referee;
A decoding algorithm $\mathsf{Dec}$ takes $n$ messages $(\mu_i)_{i\in[n]}$ as input and outputs $y=f(x_1,\ldots,x_{n})$.
The privacy requirement is that the messages $(\mu_i)_{i\in[n]}$ reveal nothing about $(x_i)_{i\in[n]}$ other than $y=f(x_1,\ldots,x_{n})$.
We define randomness complexity as the bit-length of $\rho$, denoted by $r$, and communication complexity as the total bit-length of $(\mu_i)_{i\in[n]}$, denoted by $c$.
To simplify the exposition, we here assume that $\rho$ is uniformly distributed over $\bin^r$, while our full construction allows $\rho$ to be sampled from a subset of $\bin^r$.

We follow an encoding rule $\heartsuit\,\clubsuit=1$, $\clubsuit\,\heartsuit=0$ and encode a bit string into a sequence of suits bit by bit.
For example, a string $101$ is encoded into $\heartsuit\,\clubsuit\,\clubsuit\,\heartsuit\,\heartsuit\,\clubsuit$.
Let $\Sigma$ be a PSM protocol for a function $f:\bin^n\ra\bin$.
For each $i\in[n]$, each $b\in\bin$, and each random string $\rho\in\bin^r$, we define a sequence of suits $m_i^b(\rho)$ as the encoding of a party $i$'s message $\mathsf{Enc}(i,b,\rho)$ with input $b$ and shared randomness $\rho$.
Let $L_i^b$ be a sequence of suits obtained by concatenating the $m_i^b(\rho)$'s over all possible $\rho$'s.
Here, we fix an order on the set of all $\rho$'s (e.g., a lexicographical order) and enumerate them as $\{\rho_1,\ldots,\rho_{2^r}\}$.

Suppose that parties place face-down cards whose suits correspond to $L_1^{x_1},\ldots,L_n^{x_n}$.
This sequence can be rearranged into the concatenation of $M(\rho):=(m_1^{x_1}(\rho),\ldots,m_n^{x_n}(\rho))$ over all $\rho$'s.
Note that it suffices to open only the set of cards corresponding to $M(\rho)$ for a uniformly random $\rho$, since this is equivalent to revealing the messages that the referee receives in the PSM protocol $\Sigma$.
We apply a pile-shifting shuffle to obtain a sequence of cards with suits $M(\rho_s),M(\rho_{s+1}),\ldots,M(\rho_{s-1})$, where $s$ is a uniformly random shift.
To eliminate all information except that of $M(\rho_{s})$, we then apply a complete shuffle to the remaining cards $M(\rho_{s+1}),\ldots,M(\rho_{s-1})$.
As a result, we obtain a sequence of face-down cards whose suits consist of $M(\rho_s)$ followed by a uniformly random arrangement of suits.
Finally, we turn over all cards and compute $y=f(x_1,\ldots,x_n)$ by applying the decoding algorithm of $\Sigma$ to the bit string $(\mathsf{Enc}(i,x_i,\rho_s))_{i\in[n]}$ encoded in $M(\rho_s)$.
If the communication and randomness complexity of $\Sigma$ are $c$ and $r$, respectively, then the number of cards required to encode each $M(\rho)$ is $2c$ and hence the total number of cards is $2c\cdot 2^r=c2^{r+1}$.
The state-of-the-art PSM protocol for a general function has communication and randomness complexity $O(n^32^{n/2})$ \cite{BKN18}.
Therefore, the resulting single-shuffle full-open protocol uses $2^{O(n^3 2^{n/2})}$ cards.

\subsection{Reducing the Number of Cards}
\label{section:card-reduction}

Next, we show a different construction of single-shuffle full-open protocols that uses a smaller number of cards at the cost of requiring more complicated shuffles. 

\smallskip
\noindent\textbf{Expressing General Functions by AND Functions.}
Let $f:\bin^n\ra\bin$ be a function and $f^{-1}(1)=\{\bm{a}^{(1)},\ldots,\bm{a}^{(N)}\}$, where $\bm{a}^{(j)}=(a_1^{(j)},\ldots,a_n^{(j)})\in\bin^n$ is an input on which $f$ evaluates to $1$.
For each $i\in[n]$, each $j\in[N]$, and each $x_i\in\bin$, we define $b_i^{(j)}(x_i)\in\bin$ as $b_i^{(j)}(x_i)=1$ if and only if $x_i=a_i^{(j)}$.
For any $\bm{x}=(x_1,\ldots,x_n)\in\bin^n$ with $f(\bm{x})=1$, there exists a unique $j\in[N]$ such that $\bm{x}=\bm{a}^{(j)}$; and if $f(\bm{x})=0$, then $\bm{x}\neq\bm{a}^{(j)}$ for all $j\in[N]$.
Thus, it holds that
\begin{align*}
    f(x_1,\ldots,x_n)=\sum_{j\in[N]}\mathsf{AND}_n(b_1^{(j)}(x_1),\ldots,b_n^{(j)}(x_n))
\end{align*}
for any $(x_1,\ldots,x_n)\in\bin^n$, where $\mathsf{AND}_n$ is the $n$-input AND function, i.e., $\mathsf{AND}_n(y_1,\ldots,y_n)=1$ if and only if $y_1=\cdots=y_n=1$.

Assume that we are given a single-shuffle full-open protocol $\Pi_0$ for $\mathsf{AND}_n$.
Then, we can construct a single-shuffle full-open protocol $\Pi$ for $f$ as follows.
Each party $i$ with input $x_i$ computes bits $b_i^{(j)}(x_i)$ for all $j\in[N]$.
Then, for each $j\in[N]$, parties execute $\Pi_0$ on inputs $(b_i^{(j)}(x_i))_{i\in[n]}$.
That is, each party $i$ places a sequence of face-down cards corresponding to $b_i^{(j)}(x_i)$ following the specification of $\Pi_0$, and then apply a shuffle specified by $\Pi_0$.
At this point, if $f(\bm{x})=1$, then the resulting card sequence consists of a single pile of cards corresponding to an output $1$ and $N-1$ piles of cards corresponding to an output $0$, where the position of the unique pile is determined by $j\in[N]$ such that $\bm{a}^{(j)}=\bm{x}$.
If $f(\bm{x})=0$, then the card sequence consists of $N$ piles of cards corresponding to an output $0$.
Thus, we can determine $f(\bm{x})$ from the revealed cards.
Note that before turning over cards, we need to apply a pile-scramble shuffle to ensure that the position of the pile corresponding to $1$ is independent of $\bm{x}$.
Then, the privacy of $\Pi$ directly follows from that of $\Pi_0$.
If $\Pi_0$ uses $\ell$ cards, then $\Pi$ uses $\ell N\leq \ell 2^n$ cards. 

If we instantiate $\Pi_0$ with our general protocol in the previous section, then $\ell=O(n\log n)\cdot2^{O(n\log n)}=2^{O(n\log n)}$ since there exists a PSM protocol for $\mathsf{AND}_n$ with communication and randomness complexity $O(n\log n)$ \cite{FKN94}.
Thus, the total number of cards is $2^{O(n\log n)}\cdot 2^n=2^{O(n\log n)}$, which reduces the number of cards $2^{O(n^32^{n/2})}$ in our previous protocol.
In the following, we further improve the number of cards by showing a single-shuffle full-open protocol tailored to computing $\mathsf{AND}_n$ that uses only $\ell=O(n^2)$ cards.
This leads to a protocol for general $f$ using $O(n^22^n)$ cards.

\smallskip
\noindent\textbf{A Single-Shuffle Full-Open Protocol for $\mathsf{AND}_n$.}
A key observation is that the PSM protocol for $\mathsf{AND}_n$ in \cite{FKN94} has a special algebraic structure, which we refer to as \textit{additive PSM}.
Let $m$ be a prime, $\mathbb{Z}_m$ denote the ring of integers modulo $m$, and $\mathbb{Z}_m^*$ denote its multiplicative group.
We call a PSM protocol additive over $\mathbb{Z}_m$ if shared randomness consists of an element $u$ sampled uniformly at random from a subgroup $U\subseteq\mathbb{Z}_m^*$ and additive shares $(r_i)_{i\in[n]}\in\mathbb{Z}_m^n$ of zero, which means that $(r_i)_{i\in[n]}$ are sampled uniformly at random conditioned on $\sum_{i\in[n]}r_i=0 \bmod m$.
Furthermore, a message of each party $i$ with input $x_i\in\bin$ has the form of $uy_i-r_i$, where $y_i\in\mathbb{Z}_m$ is determined by $x_i$.
The output is computed based on $\sum_{i\in[n]}(uy_i-r_i)=u\sum_{i\in[n]}y_i$.

We show a construction of single-shuffle full-open protocols from additive PSM protocols.
For several functions including $\mathsf{AND}_n$, this can improve the number of cards in our previous construction from arbitrary PSM protocols.
Our construction takes two steps:
First, we transform $n$ piles of face-down cards each encoding $y_i$ into $n$ piles each encoding $z_i:=uy_i$ for a uniformly random element $u\in U$, which is unknown to any party.
Second, we transform the resulting $n$ piles into $n$ piles each encoding $z_i-r_i \bmod m$, where $(r_i)_{i\in[n]}$ are random additive shares of zero.
Here, we use a different encoding rule: $x\in\mathbb{Z}_m$ is encoded into a sequence $\henc_m(x)$ of suits such that $\heartsuit$ is placed at $x$-th position and $\clubsuit$ is placed elsewhere, i.e.,
\begin{align*}
    \henc_m(x):=\overset{0}{\blk}\,\cdots\,\overset{x}{\red}\,\cdots\,\overset{m-1}{\blk}.
\end{align*}

To describe the first procedure, we here assume for simplicity that $U=\mathbb{Z}_m^*=\{1,2,\ldots,m-1\}$.
Since $m$ is a prime, $U$ is a cyclic group with a generator $\alpha$ and we can write $U=\{\alpha,\alpha^2,\ldots,\alpha^{m-1}\}$.
Each party with input $x_i$ computes $y_i\in\mathbb{Z}_m$ and places face-down cards with suits $\henc_m(y_i)$.
For each $\gamma\in\mathbb{Z}_m$, let $c_{i,\gamma}$ denote the $\gamma$-th suit of $\henc_m(y_i)$, that is, $c_{i,\gamma}=\heartsuit$ if and only if $\gamma=y_i$.
We rearrange the face-down cards and construct $m$ piles $p(1),\ldots,p(m-1)$, where each $p(j)$ consists of the $\alpha^j$-th cards of $\henc_m(y_i)$ for all $i\in[n]$, i.e., $p(j)=(c_{i,\alpha^j})_{i\in[n]}$.
We apply a pile-shifting shuffle to $p(1),\ldots,p(m-1)$ and denote the resulting piles by $p'(1),\ldots,p'(m-1)$.
Then, there exists a uniformly random shift $\Delta$ such that $p'(j)=p(j-\Delta)$ for all $j$, where indices are modulo $m-1$.
We define $c_{i,0}'=c_{i,0}$ and for $\gamma\neq0$, define $c_{i,\gamma}'$ as the $i$-th card of $p'(\log_\alpha\gamma)$.
Finally, we output $(c_{1,\gamma}')_{\gamma\in\mathbb{Z}_m},\ldots,(c_{n,\gamma}')_{\gamma\in\mathbb{Z}_m}$ with all cards face down.
Here, we define $\log_\alpha\gamma$ for $\gamma\in\mathbb{Z}_m\setminus\{0\}$ as $j$ such that $\gamma=\alpha^j$.
It holds that $c_{i,\gamma}'=c_{i,\alpha^{\log_\alpha\gamma-\Delta}}=c_{i,\gamma u^{-1}}$, where $u:=\alpha^{\Delta}$.
Thus, $c_{i,\gamma}'=\heartsuit$ if and only if $\gamma=y_i u$, which means that each output pile $(c_{i,\gamma}')_{\gamma\in\mathbb{Z}_m}$ is equal to $\henc_m(y_iu)$.
Since $\Delta$ is uniformly distributed over $\{1,\ldots,m-1\}$, $u$ is uniformly distributed over $U=\mathbb{Z}_m^*$.

The second procedure transforms $n$ piles of cards each with suits $\henc_m(z_i)$ into $n$ piles each with suits $\henc_m(z_i-r_i)$, where $(r_i)_{i\in[n]}$ are random additive shares of zero.
We first transform $\henc_m(z_1)$ into $\henc_m(-z_1)$ simply by reversing the order of the last $m-1$ cards:
\[
\underbrace{\overset{0}{\back}\,\overset{1}{\back}\,\overset{2}{\back}\,\cdots\,\overset{m-1}{\back}}_{\henc_{m}(z_1)}\, \rightarrow \, 
\underbrace{\overset{0}{\back}\,\overset{m-1}{\back}\,\cdots\,\overset{2}{\back}\,\overset{1}{\back}}_{\henc_{m}(-z_1)}.
\]
We then pair two sequences $\henc_m(-z_1)$ and $\henc_m(z_n)$, and apply a pile-shifting shuffle:
\[
\begin{tabular}{c}
$\overbrace{\back \, \back \, \cdots \, \back}^{\henc_{m}(-z_1)}$\vspace{1mm}\\
$\underbrace{\back \, \back \, \cdots \, \back}_{\henc_{m}(z_n)}$
\end{tabular}
\ra
\begin{tabular}{c}
$\overbrace{\back \, \back \, \cdots \, \back}^{\henc_{m}(-z_1+r_1)}$\vspace{1mm}\\
$\underbrace{\back \, \back \, \cdots \, \back}_{\henc_{m}(z_n+r_1)}$
\end{tabular}
\]
Here, $r_1$ is a uniformly random shift.
Finally, we reverse the order of the last $m-1$ cards of $\henc_m(-z_1+r_1)$.
As a result, we obtain two sequences $\henc_m(z_1-r_1)$ and $\henc_m(z_n+r_1)$.
We apply the same procedure to a pair of $\henc_m(z_2)$ and $\henc_m(z_n+r_1)$, and obtain $\henc_m(z_2-r_2)$ and $\henc_m(z_n+r_1+r_2)$, where $r_2$ is another uniformly random shift.
By continuing these procedures, we finally obtain $\henc_m(z_1-r_1),\ldots,\henc_m(z_n-r_n)$, where $r_1,\ldots,r_{n-1}$ are uniformly at random and $r_n=-\sum_{i=1}^{n-1}r_i$.

Now, we can transform an additive PSM protocol $\Sigma$ for $f$ into a single-shuffle full-open protocol $\Pi$ for $f$.
After performing the above two procedures, we obtain face-down cards with suits $\henc_m(uy_1-r_1),\ldots,\henc_m(uy_n-r_n)$, where $u$ is a uniformly random element of $U$ and $(r_i)_{i\in[n]}$ are random additive shares of zero.
We turn over all cards and run the decoding algorithm of $\Sigma$ on $(uy_i-r_i)_{i\in[n]}$.
The correctness of $\Sigma$ ensures that the output is $f(x_1,\ldots,x_n)$ and the privacy of $\Sigma$ ensures that the revealed suits leak nothing but $f(x_1,\ldots,x_n)$.
The number of cards in $\Pi$ is $mn$.
Since $\mathsf{AND}_n$ admits an additive PSM protocol over $\mathbb{Z}_m$ such that $m=O(n)$ \cite{FKN94}, we obtain a single-shuffle full-open protocol for $\mathsf{AND}_n$ using $O(n^2)$ cards.

\smallskip
\noindent\textbf{Putting It All Together.}
By plugging the protocol for $\mathsf{AND}_n$ into the above construction, we obtain a single-shuffle full-open protocol for any function $f:\bin^n\ra\bin$ using at most $O(n^2)\cdot 2^n=O(n^22^n)$ cards.
This improves upon our previous construction, which requires $2^{O(n^32^{n/2})}$ cards.
However, it comes at the cost of more complicated shuffles: this protocol requires a composition of $O(n2^n)$ pile-shifting shuffles followed by one pile-scramble shuffle, while our previous protocol needs only one pile-shifting shuffle and one complete shuffle.

\section{Preliminaries}

\smallskip
\noindent\textbf{Notations.}
For a natural number $n\in\mathbb{N}$, we denote $[n]=\{1,2,\ldots,n\}$.
Let $\mathfrak{S}_N$ denote the set of all permutations over $[N]$.
For a probability distribution $\mathcal{D}$, we write $\rho\la\mathcal{D}$ if an element $\rho$ is sampled according to $\mathcal{D}$.
We denote the support of $\mathcal{D}$, i.e., the set of all possible outcomes, by $\mathrm{supp}(\mathcal{D})$.
For a set $X$, we write $\rho\sample X$ if an element $\rho$ is sampled uniformly at random from $X$.
For an integer $m\geq 2$, we denote the ring of integers modulo $m$ by $\mathbb{Z}_m$ and identify it with $\{0,1,\ldots,m-1\}$ as a set.
Let $\mathbb{Z}_m^*$ denote the multiplicative group of integers modulo $m$ and naturally identify it with a subset of $\{0,1,\ldots,m-1\}$.
We call a function $f:\bin^n\ra\bin$ symmetric if $f(x_{\sigma(1)},\ldots,x_{\sigma(n)})=f(x_1,\ldots,x_n)$ for any input $(x_i)_{i\in[n]}\in\bin^n$ and any permutation $\sigma\in\mathfrak{S}_n$.
We define the $n$-input AND function $\mathsf{AND}_n:\bin^n\ra\bin$ as $\mathsf{AND}_n(x_1,\ldots,x_n)=1$ if and only if $x_1=\cdots=x_n=1$.

\subsection{Card-Based Protocols}

\smallskip
\noindent\textbf{Cards.}
In this paper, we use \textit{binary cards} whose front sides are either $\blk$ or $\red$ and back sides are both $\back\,$. 
We assume that two cards with the same suit are indistinguishable. 
For $x\in\mathbb{Z}_m$, let $\henc_{m}(x)\in\{\heartsuit,\clubsuit\}^m$ denote the sequence of suits such that $\heartsuit$ is placed at $x$-th position and $\clubsuit$ is placed elsewhere.

\smallskip
\noindent\textbf{Shuffles.}
A shuffle is an operation that applies a random permutation to a sequence of $\ell$ face-down cards, where the permutation is chosen from some probability distribution over $\mathfrak{S}_\ell$.
We allow a shuffle to be applied not only to the entire sequence of cards but also to a subsequence of cards.
We assume that no party guesses which permutation is chosen during the shuffle. 
Below, we introduce basic shuffles.

A \textit{complete shuffle} is a shuffle that applies a uniformly random permutation to a sequence of $\ell$ face-down cards, which is denoted by $[\cdot]$. 
For example, a complete shuffle for a sequence of three cards results in one of the six sequences each with probability $1/6$ as follows:
\[
\Bigl[ \; \overset{1}{\back} \, \overset{2}{\back} \, \overset{3}{\back} \; \Bigr]
~\rightarrow~
\overset{1}{\back} \, \overset{2}{\back} \, \overset{3}{\back} ~\mbox{or}~ \overset{1}{\back} \, \overset{3}{\back} \, \overset{2}{\back} ~\mbox{or}~
\overset{2}{\back} \, \overset{3}{\back} \, \overset{1}{\back} ~\mbox{or}~ \overset{2}{\back} \, \overset{1}{\back} \, \overset{3}{\back} ~\mbox{or}~
\overset{3}{\back} \, \overset{1}{\back} \, \overset{2}{\back} ~\mbox{or}~ \overset{3}{\back} \, \overset{2}{\back} \, \overset{1}{\back}\,.
\]

A \emph{pile-scramble shuffle} \cite{IshikawaUCNC2015} is a shuffle that divides a sequence of cards into multiple \emph{piles} of the same number of cards and applies a random permutation to a sequence of piles, which is denoted by $[\cdot | \cdots | \cdot]$. 
For example, a pile-scramble shuffle for a sequence of three piles each having two cards results in one of the six sequences each with probability $1/6$ as follows:
\[
\Bigl[ \; \overset{1}{\back} \, \overset{2}{\back} \;\Big|\; \overset{3}{\back} \, \overset{4}{\back} \;\Big|\; \overset{5}{\back} \, \overset{6}{\back} \; \Bigr]
~\rightarrow~
\begin{cases}
~\overset{1}{\back} \, \overset{2}{\back} ~~ \overset{3}{\back} \, \overset{4}{\back} ~~ \overset{5}{\back} \, \overset{6}{\back}\\
~\overset{1}{\back} \, \overset{2}{\back} ~~ \overset{5}{\back} \, \overset{6}{\back} ~~ \overset{3}{\back} \, \overset{4}{\back}\\
~\overset{3}{\back} \, \overset{4}{\back} ~~ \overset{5}{\back} \, \overset{6}{\back} ~~ \overset{1}{\back} \, \overset{2}{\back}\\
~\overset{3}{\back} \, \overset{4}{\back} ~~ \overset{1}{\back} \, \overset{2}{\back} ~~ \overset{5}{\back} \, \overset{6}{\back}\\
~\overset{5}{\back} \, \overset{6}{\back} ~~ \overset{1}{\back} \, \overset{2}{\back} ~~ \overset{3}{\back} \, \overset{4}{\back}\\
~\overset{5}{\back} \, \overset{6}{\back} ~~ \overset{3}{\back} \, \overset{4}{\back} ~~ \overset{1}{\back} \, \overset{2}{\back}
\end{cases}.
\]

A \textit{random cut} is a shuffle that applies a cyclic shift to a sequence of cards a uniformly random number of times, which is denoted by $\langle \cdot \rangle$. For example a random cut for a sequence of four cards results in one of the four sequences each with probability $1/4$ as follows:
\[
\Bigl\langle \; \overset{1}{\back} \, \overset{2}{\back} \, \overset{3}{\back} \, \overset{4}{\back} \;\Bigr\rangle
~\rightarrow~
\overset{1}{\back} \, \overset{2}{\back} \, \overset{3}{\back} \, \overset{4}{\back} ~\mbox{or}~
\overset{2}{\back} \, \overset{3}{\back} \, \overset{4}{\back} \, \overset{1}{\back} ~\mbox{or}~ \overset{3}{\back} \, \overset{4}{\back} \, \overset{1}{\back} \, \overset{2}{\back} ~\mbox{or}~
\overset{4}{\back} \, \overset{1}{\back} \, \overset{2}{\back} \, \overset{3}{\back}\,.
\]

A \textit{pile-shifting shuffle} is a shuffle that divides a sequence of cards into multiple piles of the same number of cards and applies a cyclic shift to a sequence of cards a uniformly random number of times, which is denoted by $\langle \cdot | \cdots | \cdot \rangle$. 
For example, a pile-shifting shuffle for a sequence of three piles each having two cards results in one of the three sequences each with probability $1/3$ as follows:
\[
\Bigl\langle \; \overset{1}{\back} \, \overset{2}{\back} \;\Big|\; \overset{3}{\back} \, \overset{4}{\back} \;\Big|\; \overset{5}{\back} \, \overset{6}{\back} \; \Bigr\rangle
~\rightarrow~
\overset{1}{\back} \, \overset{2}{\back} \, \overset{3}{\back} \, \overset{4}{\back} \, \overset{5}{\back} \, \overset{6}{\back} ~\mbox{or}~
\overset{3}{\back} \, \overset{4}{\back} \, \overset{5}{\back} \, \overset{6}{\back} \, \overset{1}{\back} \, \overset{2}{\back} ~\mbox{or}~ 
\overset{5}{\back} \, \overset{6}{\back} \, \overset{1}{\back} \, \overset{2}{\back} \, \overset{3}{\back} \, \overset{4}{\back}\,.
\]

A complete shuffle is preferable to a pile-scramble shuffle since the latter requires envelopes or sleeves to make piles.
For the same reason, a random cut is preferable to a pile-shifting shuffle.
We consider that a complete shuffle and a random cut are incomparable.

A \textit{composition} of two shuffles is an operation that applies two shuffles to a sequence consecutively. 
We regard any composition of shuffles as a single shuffle.
Formally, let $\mathcal{F}_1, \mathcal{F}_2$ be the probability distributions of the first and the second shuffles, respectively. 
Then the probability distribution $\mathcal{F}$ of the composed shuffle is defined as
\[
\Pr[\sigma \leftarrow \mathcal{F}] = \sum_{\substack{\sigma_1, \sigma_2:\\\,\sigma_2 \sigma_1 = \sigma}} \Pr[\sigma_1 \leftarrow \mathcal{F}_1] \cdot \Pr[\sigma_2 \leftarrow \mathcal{F}_2].
\]

\smallskip
\noindent\textbf{Protocols.}
We follow the formalization of card-based protocols in \cite{MizukiIJIS2014}.
Since we focus on protocols with a single shuffle, the formalization is simplified as follows.

\begin{definition}\label{definition:ssfo}
    A \textup{single-shuffle card-based protocol $\Pi$} for a function $f:\bin^n\ra\bin$ is specified by $2n$ sequences of suits $(L_i^b)_{i\in[n],b\in\bin}$ with $L_i^b\in\{\heartsuit,\clubsuit\}^{\ell_i}$, a probability distribution $\mathcal{D}$ over $\mathfrak{S}_{\ell}$, and an index set $R\subseteq[\ell]$ of size $k$, where $\ell=\sum_{i\in[n]}\ell_i$.
    For an input $\bm{x}=(x_1,\ldots,x_n)\in\bin^n$, a final state $\mathsf{st}_{\bm{x}}\in\{\heartsuit,\clubsuit\}^k$ of $\Pi$ is defined as a sequence of suits obtained through the following procedures:
    \begin{enumerate}
        \item Prepare a sequence of $\ell$ face-down cards with suits $(L_1^{x_1},\ldots,L_n^{x_n})$.
        \item Sample a permutation $\pi\leftarrow\mathcal{D}$ and permute the sequence of the face-down cards according to $\pi$.
        \item Reveal the cards at positions specified by the index set $R$.
        \item Let $\mathsf{st}_{\bm{x}}$ be the sequence of suits of the revealed cards.
    \end{enumerate}
    $\Pi$ is required to satisfy the following properties:
    \begin{description}
        \item[Correctness.]
        There exists a function $\mathsf{Output}:\{\heartsuit,\clubsuit\}^k\ra\bin$ such that for any input $\bm{x}\in\bin^n$, it holds that $\prob{\mathsf{Output}(\mathsf{st}_{\bm{x}})=f(\bm{x})}=1$, where the randomness is taken over the choice of $\pi\la\mathcal{D}$.
        \item[Privacy.]
        For any pair of inputs $\bm{x},\bm{x}'\in\bin^n$ with $f(\bm{x})=f(\bm{x}')$, the distribution of $\mathsf{st}_{\bm{x}}$ is identical to that of $\mathsf{st}_{\bm{x}'}$.
    \end{description}
    We define the \textup{number of cards} of $\Pi$ as the minimum value of $\ell_{\heartsuit}+\ell_{\clubsuit}$ over all $(\ell_{\heartsuit},\ell_{\clubsuit})$'s such that for any $(x_i)_{i\in[n]}$, the sequence of suits $(L_1^{x_1},\ldots,L_n^{x_n})$ contains at most $\ell_{\heartsuit}$ $\heartsuit$'s and $\ell_{\clubsuit}$ $\clubsuit$'s\footnote{In all protocols presented in the paper, $L_i^0$ contains the same number of $\heartsuit$ (and hence the same number of $\clubsuit$) as $L_i^1$. In this case, the number of cards of $\Pi$ equals $\ell=\sum_{i\in[n]}\ell_i$.}.
    The \textup{shuffle} of $\Pi$ refers to the distribution $\mathcal{D}$.
\end{definition}

If $R=[\ell]$, then we say that $\Pi$ is a single-shuffle \textit{full-open} protocol.
If $R$ is a strict subset of $[\ell]$, then we call $\Pi$ a single-shuffle \textit{static-open} protocol.
We also introduce a single-shuffle \textit{adaptive-open} protocol, in which the position of the next card to be turned face up is determined based on the suits of the cards that have already been revealed.
Formally, an adaptive-open protocol additionally specifies a function that takes a sequence of pairs $(p_i,s_i)\in[\ell]\times\{\heartsuit,\clubsuit\}$ of a position and a suit as input, and outputs the next position $p\in[\ell]$.

\subsection{Private Simultaneous Messages Protocols}

A Private Simultaneous Messages (PSM) protocol \cite{FKN94,IK97} is a cryptographic primitive that enables $n$ input parties each holding an input $x_i$ to reveal $f(x_1,\ldots,x_{n})$ to an external party without revealing any extra information.
Below, we give the formal definition.

\begin{definition}
    A \textup{Private Simultaneous Messages protocol} $\Sigma$ (or a \textup{PSM protocol} for short) for a function $f:\bin^n\rightarrow\bin$ is specified by three algorithms $\mathsf{Gen}$, $\mathsf{Enc}$, and $\mathsf{Dec}$, where:
    \begin{itemize}
        \item $\mathsf{Gen}()\ra \rho$: $\mathsf{Gen}$ is a randomized algorithm that takes no input and outputs randomness $\rho\la\mathcal{D}$, where $\mathcal{D}$ is a probability distribution over $\bin^r$.
        \item $\mathsf{Enc}(i,x_i,\rho)\ra \mu_i$: $\mathsf{Enc}$ is a deterministic algorithm that takes $i\in[n]$, $x_i\in\bin$, and $\rho\in\bin^r$ as input and outputs a message $\mu_i\in\bin^{c_i}$.
        \item $\mathsf{Dec}(\mu_1,\ldots,\mu_n)\ra y$: $\mathsf{Dec}$ is a deterministic algorithm that takes $n$ messages $(\mu_i)_{i\in[n]}$ as input and outputs $y\in\bin$.
    \end{itemize}
    satisfying the following properties:
    \begin{description}
        \item[Correctness.] For any $\bm{x}=(x_i)_{i\in[n]}$, it holds that
        \begin{align*}
            \prob{\mathsf{Dec}((\mathsf{Enc}(i,x_i,\rho))_{i\in[n]})=f(x_1,\ldots,x_{n})}=1,
        \end{align*}
        where the probability is taken over the choice of $\rho\la\mathsf{Gen}()$.
        \item[Privacy.] For any $\bm{x}=(x_i)_{i\in[n]},\bm{x}'=(x_i')_{i\in[n]}$ with $f(\bm{x})=f(\bm{x}')$, the distribution of $(\mathsf{Enc}(i,x_i,\rho))_{i\in[n]}$ induced by $\rho\la\mathsf{Gen}()$ is identical to that of $(\mathsf{Enc}(i,x_i',\rho'))_{i\in[n]}$ induced by $\rho'\la\mathsf{Gen}()$.
    \end{description}
    We define the \textup{randomness complexity} of $\Sigma$ as $r$ and the \textup{communication complexity} of $\Sigma$ as $\sum_{i\in[n]}c_i$.
\end{definition}

By default, we suppose that $\mathcal{D}$ is a uniform distribution over a subset of $\bin^r$.

The state-of-the-art PSM protocol for any function is given by \cite{BKN18}.

\begin{proposition}[Beimel et al. \cite{BKN18}]\label{proposition:psm-all}
For any function $f:\bin^n\ra\bin$, there exists a PSM protocol for $f$ with communication and randomness complexity $O(n^3 2^{n/2})$.
\end{proposition}

For specific functions, PSM protocols with better communication and randomness complexity are known \cite{FKN94,Ish13,SESN23}.
These protocols have a common template.
We refer to such special protocols as \textit{additive} PSM protocols.
Let $m$ be a prime.
We say that a PSM protocol $\Sigma=(\mathsf{Gen},\mathsf{Enc},\mathsf{Dec})$ is additive over $\mathbb{Z}_m$ if there are a subgroup $U\subseteq\mathbb{Z}_m^*$ and $n$ functions $(g_i:\bin\ra\mathbb{Z}_m)_{i\in[n]}$ such that
\begin{itemize}
    \item Randomness generated by $\mathsf{Gen}$ is $\rho=(u,(r_i)_{i\in[n]})$, where $u$ is sampled according to the uniform distribution over $U$ and $(r_i)_{i\in[n]}$ are random additive shares of zero (i.e., uniformly random elements conditioned on $\sum_{i\in[n]}r_i=0$). 
    \item Each message has the form of $\mu_i=u\cdot g_i(x_i)-r_i \bmod m$.
    \item $\mathsf{Dec}$ computes an output based on $\sum_{i\in[n]}\mu_i \bmod m$.
\end{itemize}
The communication and randomness complexity is $O(n\log m)$.

\begin{proposition}[Feige et al. \cite{FKN94}]\label{proposition:psm-and}
There exists an additive PSM protocol over $\mathbb{Z}_m$ for $\mathsf{AND}_n$ such that $m$ is any prime larger than $n$ and $U=\mathbb{Z}_m^*$.
In particular, the communication and randomness complexity is $O(n\log n)$.
\end{proposition}

\begin{proposition}[Shinagawa et al. \cite{SESN23}]\label{proposition:psm-sym}
For any symmetric function $f:\bin^n\ra\bin$, there exists an additive PSM protocol over $\mathbb{Z}_m$ for $f$ such that $m=(1+o(1))n^2 2^{2n-2}$ for any sufficiently large $n$ and $U\subseteq\mathbb{Z}_m^*$ is the subgroup of all quadratic residues $a$ modulo $m$, i.e., $a\in\mathbb{Z}_m^*$ for which there exists $x\in\mathbb{Z}_m^*$ such that $x^2=a \bmod m$.
In particular, the communication and randomness complexity is $O(n^2)$.
\end{proposition}

\section{Single-Shuffle Protocols for Any Function}

First, we show a transformation from a PSM protocol for a function $f$ into a single-shuffle full-open protocol for $f$ using a composition of a pile-shifting shuffle and a complete shuffle.

\begin{figure}
\centering
\fbox{\parbox{\textwidth}{
\begin{description}
    \item[Notations.]\mbox{}
        \begin{itemize}       
            \item Let $f:\bin^n\ra\bin$ be a function.
            \item Let $\Sigma=(\mathsf{Gen},\mathsf{Enc},\mathsf{Dec})$ be a PSM protocol for $f$ with randomness complexity $r$ and communication complexity $c$. 
            \item Let $\mathrm{supp}(\mathcal{D})=\{\rho_1,\ldots,\rho_{R}\}$ be an enumeration of all possible values of randomness for $\Sigma$.
        \end{itemize}

    \item[Construction of $L_i^b$.] For each $i\in[n]$ and each $b\in\bin$,
    \begin{enumerate}
        \item Denote $\mu_i(b,\rho)=\mathsf{Enc}(i,b,\rho)\in\bin^{c_i}$ for each $\rho\in\mathrm{supp}(\mathcal{D})$.
        \item Let $\widetilde{L}_i^b=\mu_i(b,\rho_1)||\cdots||\mu_i(b,\rho_{R})\in\bin^{c_iR}$, where $||$ denotes the concatenation of strings.
        \item Let $L_i^b\in\{\heartsuit,\clubsuit\}^{\ell_i}$ be the sequence of suits determined by $\widetilde{L}_i^b$ based on the encoding rule $\clubsuit\heartsuit= 0$ and $\heartsuit\clubsuit= 1$, where $\ell_i=2c_iR$.
    \end{enumerate}

    \item[Shuffling.] Given a sequence of $\ell=\sum_{i\in[n]}\ell_i$ face-down cards
    \begin{align*}
    \overbrace{\underbrace{\back\cdots\back}_{\mu_1(x_1,\rho_1)}\cdots\underbrace{\back\cdots\back}_{\mu_1(x_1,\rho_{R})}}^{L_1^{x_1}}\cdots\overbrace{\underbrace{\back\cdots\back}_{\mu_n(x_n,\rho_1)}\cdots\underbrace{\back\cdots\back}_{\mu_n(x_n,\rho_{R})}}^{L_n^{x_n}}\,,
    \end{align*}
    \begin{enumerate}
        \item Apply a pile-shifting shuffle as follows:
        \begin{align*}
            &\Big\langle
            \underbrace{\back\cdots\back}_{\mu_1(x_1,\rho_1)}\cdots\underbrace{\back\cdots\back}_{\mu_n(x_n,\rho_{1})}
            \Big|\cdots\Big|
            \underbrace{\back\cdots\back}_{\mu_1(x_1,\rho_{R})}\cdots\underbrace{\back\cdots\back}_{\mu_n(x_n,\rho_{R})}
            \Big\rangle\\
            \ra~&
            \underbrace{\back\cdots\back}_{\mu_1(x_1,\rho_s)}\cdots\underbrace{\back\cdots\back}_{\mu_n(x_n,\rho_{s})}            \underbrace{\back\cdots\back}_{\mu_1(x_1,\rho_{s+1})}\cdots\underbrace{\back\cdots\back}_{\mu_n(x_n,\rho_{s+1})}
            \cdots
            \underbrace{\back\cdots\back}_{\mu_1(x_1,\rho_{s-1})}\cdots\underbrace{\back\cdots\back}_{\mu_n(x_n,\rho_{s-1})},
        \end{align*}
        where $s\in[R]$.
        \item Apply a complete shuffle to the sequence of cards, except for those corresponding to $(\mu_1(x_1,\rho_s),\ldots,\mu_n(x_n,\rho_{s}))$. That is,
        \begin{align*}
            &\underbrace{\back\cdots\back}_{\mu_1(x_1,\rho_s)}\cdots\underbrace{\back\cdots\back}_{\mu_n(x_n,\rho_{s})}\,
            \Big[\underbrace{\back\cdots\back}_{\mu_1(x_1,\rho_{s+1})}\cdots\underbrace{\back\cdots\back}_{\mu_n(x_n,\rho_{s+1})}
            \cdots
            \underbrace{\back\cdots\back}_{\mu_1(x_1,\rho_{s-1})}\cdots\underbrace{\back\cdots\back}_{\mu_n(x_n,\rho_{s-1})}\Big]\\
            \ra~&
            \underbrace{\back\cdots\back}_{\mu_1(x_1,\rho_s)}\cdots\underbrace{\back\cdots\back}_{\mu_n(x_n,\rho_{s})}\,\back\cdots\cdots\back\,.
        \end{align*}
    \end{enumerate}

    \item[Output.]\mbox{}
    \begin{enumerate}
        \item Reveal all face-down cards.
        \item Run $\mathsf{Dec}$ on input $(\mu_1(x_1,\rho_s),\ldots,\mu_n(x_n,\rho_{s}))$ to obtain $y\in\bin$.
        \item Output $y$.
    \end{enumerate}
\end{description}}}
\caption{A single-shuffle full-open protocol based on a PSM protocol}
\label{figure:full-open-1}
\end{figure}

\begin{theorem}\label{theorem:full-open-1}
    Let $\Sigma=(\mathsf{Gen},\mathsf{Enc},\mathsf{Dec})$ be a PSM protocol for a function $f:\bin^n\ra\bin$ with randomness complexity $r$ and communication complexity $c$.
    Then, there exists a single-shuffle full-open protocol $\Pi$ for $f$ such that
    \begin{itemize}
        \item The number of cards is at most $c2^{r+1}$.
        \item The shuffle is a composition of a pile-shifting shuffle and a complete shuffle.
    \end{itemize}
\end{theorem}

\begin{proof}
Let $\Pi$ be a card-based protocol described in Figure~\ref{figure:full-open-1}.
The correctness of $\Pi$ follows straightforwardly from that of the PSM protocol $\Sigma$ since 
\begin{align*}
    \mathsf{Dec}((\mu_i(x_i,\rho_s))_{i\in[n]})=\mathsf{Dec}((\mathsf{Enc}(i,x_i,\rho_s))_{i\in[n]})=f((x_i)_{i\in[n]}).
\end{align*}
To see privacy, let $\bm{x}=(x_i)_{i\in[n]}$ and $\bm{x}'=(x_i')_{i\in[n]}$ be inputs with $f(\bm{x})=f(\bm{x}')$.
Since $\rho_s$ at Step~1 of the shuffling phase is uniformly sampled from $\{\rho_1,\ldots,\rho_{R}\}$, the privacy of $\Sigma$ ensures that the distribution of $(\mu_i(x_i,\rho_s))_{i\in[n]}$ is identical to that of $(\mu_i(x_i',\rho_s))_{i\in[n]}$.
Furthermore, a complete shuffle at Step~2 ensures that the suits except for those corresponding to $(\mu_i(x_i,\rho_s))_{i\in[n]}$ or $(\mu_i(x_i',\rho_s))_{i\in[n]}$ are equally distributed, independent of $\bm{x}$ or $\bm{x}'$.
Therefore, we conclude that the distribution of suits revealed in the output phase on input $\bm{x}$ is identical to that on input $\bm{x}'$.

Since we use the encoding rule $\heartsuit\clubsuit=1$ and $\clubsuit\heartsuit=0$, the number of cards to encode $(\mu_i(x_i,\rho))_{i\in[n]}$ is $2c$.
Thus, the total number of cards is $2cR\leq c2^{r+1}$.
\end{proof}

It is possible to extend the above construction to the setting where $\mathcal{D}$ is an arbitrary distribution over $\bin^r$.
At Step~1 in the shuffling phase, we sample a non-uniform random shift $s\in[R]$ according to $\mathcal{D}$, that is, $s$ is sampled with probability $\probsub{\mathcal{D}}{\rho_s}$.

Next, we show a variant of the construction in Theorem~\ref{theorem:full-open-1}, which uses a simpler shuffle at the price of giving up full-opening.
Consider a modification of the construction in which only the pile of cards corresponding to $p_s:=(\mu_i(x_i,\rho_s))_{i\in[n]}$ are revealed in the output phase.
Then, it is no longer necessary to apply a complete shuffle to the remaining cards in the shuffling phase.
Since the positions of cards to be revealed are fixed, this leads to a single-shuffle static-open protocol that only requires a pile-shifting shuffle.
Furthermore, we observe that the pile-shifting shuffle is applied to the sequence to select a pile $p_s$ uniformly at random from the sequence of piles $(p_1, \ldots, p_{2^r})$.
We refer to this functionality as \textit{uniform selection}.
We will show that with some additional cards, uniform selection can be implemented using a random cut only.
As a result, we obtain the following theorem.

\begin{theorem}\label{theorem:full-open-random-cut-static}
    Let $\Sigma$ be a PSM protocol for a function $f:\bin^n\ra\bin$ with randomness complexity $r$ and communication complexity $c$.
    Then, there exists a single-shuffle static-open protocol $\Pi$ for $f$ such that
    \begin{itemize}
        \item The number of cards is $c2^{r+1}+2^{r+2}$.
        \item The shuffle is a random cut.
    \end{itemize}
\end{theorem}

To prove Theorem~\ref{theorem:full-open-random-cut-static}, we begin with showing a card-based protocol that realizes uniform selection using a random cut.

\begin{lemma}\label{lemma:uniform_selection-static}
Let $(p_0, p_1, \ldots, p_{R-1})$ be a sequence of $R$ piles of cards, each consisting of the same number of cards. 
Suppose that $p_i$ does not contain $\ell$ consecutive $\heartsuit$s as a subsequence for any $i\in\mathbb{Z}_R$. 
Then, there exists a static-open protocol that implements a uniform selection from $(p_0, p_1, \ldots, p_{R-1})$ by using a random cut and $(\ell+1)R$ additional cards. 
\end{lemma}

\begin{proof}
The protocol additionally uses $\ell R$ $\red\,$s and $R$ $\blk\,$s and makes $R$ piles each consisting of $\ell$ $\red\,$s followed by one $\blk\,$.
We call each such pile a \textit{delimiter}.
Let $\widetilde{p}_i$ denote a pile consisting of $p_i$ followed by a delimiter.
Let $C$ be the number of cards in each $\widetilde{p}_i$.
The protocol first constructs the following sequence:
\[
\overbrace{\underbrace{\back\,\cdots\,\back}_{p_0}~\underbrace{\red\,\cdots\,\red}_{\ell~\text{cards}}\,\blk}^{\widetilde{p}_0}~
\overbrace{\underbrace{\back\,\cdots\,\back}_{p_1}~\underbrace{\red\,\cdots\,\red}_{\ell~\text{cards}}\,\blk}^{\widetilde{p}_1}~
\cdots~
\overbrace{\underbrace{\back\,\cdots\,\back}_{p_{R-1}}~\underbrace{\red\,\cdots\,\red}_{\ell~\text{cards}}\,\blk\,}^{\widetilde{p}_{R-1}}.
\]
For each $j\in\mathbb{Z}_R$ and $k\in\mathbb{Z}_C$, let $a_{j,k}$ denote the $k$-th card in the pile $\widetilde{p}_j$. (Here, we number positions starting at $0$.)
Next, we rearrange the above sequence into $b_0,b_1,\ldots,b_{N-1}$, where $N=RC$ and 
\begin{align*}
    b_i=a_{j,k}~\text{with}~j=(i\bmod R)~\text{and}~k=\lfloor i/R \rfloor
\end{align*}
for each $i\in\mathbb{Z}_N$.
The protocol applies a random cut to the sequence $b_0,b_1,\ldots,b_{N-1}$.
Let $b_0',b_1',\ldots,b_{N-1}'$ denote the resulting sequence.
We open $C$ cards at positions $hR$ for $h\in\mathbb{Z}_C$, i.e., $b_0',b_R',\ldots,b_{(C-1)R}'$.

It holds that $b_i'=b_{s-i}$ for every $i$ where $s\in\mathbb{Z}_N$ is a uniformly random integer sampled by the random cut and the index is computed modulo $N$.
Thus, the opened cards are 
\begin{align*}
    b_{hR}'=b_{s-hR}=a_{j_h,k_h}~(\forall h\in\mathbb{Z}_C),
\end{align*}
where
\begin{align*}
    j_h=(s-hR\bmod R)=(s\bmod R)~\text{and}~
    k_h=\lfloor(s-hR)/R\rfloor=\lfloor s/R\rfloor-h.
\end{align*}
The suits of these cards are equal to those of 
\begin{align*}
    \widetilde{p}_{s\bmod R}=(a_{s\bmod R,0},\ldots,a_{s\bmod R,C-1})
\end{align*}
shifted by $\lfloor s/R\rfloor$.
Since no $p_i$ contains $\ell$ consecutive $\heartsuit$s, the delimiter enables us to identify the head of $p_{s\bmod R}$ from the revealed suits.
Namely, its head starts after the delimiter $\red\,\cdots\,\red\,\blk\,$.

It remains to show that the index $s\bmod R$ of the revealed pile is uniformly distributed over $\mathbb{Z}_R$ even conditioned on the shift $\lfloor s/R\rfloor$.
There is a one-to-one correspondence between $\mathbb{Z}_N$ and $\mathbb{Z}_R\times\mathbb{Z}_C$, that is, $(s\bmod R,\lfloor s/R\rfloor)$ uniquely determines $s$.
Thus, if $s$ is uniformly distributed over $\mathbb{Z}_N$, then for any $j\in\mathbb{Z}_R$ and $k\in\mathbb{Z}_C$, it holds that
\begin{align*}
    \condprob{s\bmod R=j}{\lfloor s/R\rfloor=k}
    =\frac{1/N}{R/N}
    =\frac{1}{R}.
\end{align*}
Therefore, the above procedure implements a uniform selection. 
\end{proof}

\begin{proof}[Proof of Theorem~\ref{theorem:full-open-random-cut-static}]
Since the encoding rule $\clubsuit\heartsuit= 0,\heartsuit\clubsuit= 1$ ensures that each pile $p_s$ does not contain $\ell=3$ consecutive $\heartsuit$s, uniform selection from $(p_1,\ldots,p_{R})$ can be implemented with $4R\leq 4 \cdot 2^{r} = 2^{r+2}$ additional cards.
We thus obtain a single-shuffle static-open protocol using a random cut and at most $2^{r+2}$ additional cards, which completes the proof of Theorem~\ref{theorem:full-open-random-cut-static}.
\end{proof}

We apply the above theorems to the PSM protocol in Proposition~\ref{proposition:psm-all} and obtain the following corollary.

\begin{corollary}\label{corollary:SSFO-all}
For any function $f:\bin^n\ra\bin$, there exist single-shuffle protocols $\Pi_1$ and $\Pi_2$ for $f$ such that
\begin{itemize}
    \item $\Pi_1$ is full-open and uses a composition of a pile-shifting shuffle and a complete shuffle;
    \item $\Pi_2$ is static-open and uses a random cut;
\end{itemize}
and both of them require $2^{O(n^3 2^{n/2})}$ cards.
\end{corollary}

\section{Single-Shuffle Full-Open Protocols with Fewer Cards}

We show a different single-shuffle full-open protocol reducing the number of cards from $2^{O(n^3 2^{n/2})}$ to $O(n^22^n)$ at the cost of using a composition of $O(n2^n)$ shuffles.

\subsection{A Construction Based on Additive PSM}

First, we present a transformation from additive PSM protocols to single-shuffle full-open protocols.
Recall that in an additive PSM protocol over $\mathbb{Z}_m$, each message has the form of $uy_i-r_i \bmod m$, where $u$ is sampled uniformly at random from a subgroup $U\subseteq\mathbb{Z}_m^*$, $y_i=g_i(x_i)$ is determined by an input $x_i$, and $(r_i)_{i\in[n]}$ are random additive shares of zero.

The following proposition shows a card-based protocol for multiplying each $y_i$ by a common $u$.

\begin{proposition}\label{proposition:full-open-additive-PSM-1}
There exists a single-shuffle protocol such that for any $(y_i)_{i\in[n]}\in\mathbb{Z}_m^n$, a sequence of face-down cards with suits $(\henc_m(y_1),\ldots,\henc_m(y_n))$ is transformed into a sequence with suits $(\henc_m(uy_1),\ldots,\henc_m(uy_n))$, where $u\sample U$.
The procedures involve one pile-shifting shuffle. 
\end{proposition}

\begin{figure}
\centering
\fbox{\parbox{\textwidth}{
\begin{description}
    \item[Notations.]\mbox{}
        \begin{itemize}       
            \item Let $m$ be a prime.
            \item Let $U = \{\alpha_0,\ldots,\alpha_{k-1}\}$ be a subgroup of $\mathbb{Z}_m^*$. Since $U$ is necessarily a cyclic group, we may assume that $\alpha_0$ is a generator and $\alpha_\ell=\alpha_0^{\ell+1}$ for any $\ell\in\{0,1,\ldots,k-1\}$.
            \item Let $\beta_1,\ldots,\beta_t\in\mathbb{Z}_m^*$ be coset representatives with respect to $U$, i.e., $\mathbb{Z}_m^*=\bigcup_{s\in[t]}\beta_s U$, where $t=(m-1)/k$.
        \end{itemize}

    \item[Input.] $mn$ face-down cards with suits $(\henc_m(y_1),\ldots,\henc_m(y_n))$, where $y_i\in\mathbb{Z}_m$ and $\henc_m(y_i)=(c_{i,0},\ldots,c_{i,m-1})\in\{\heartsuit,\clubsuit\}^m$.
    
    \item[Output.] $mn$ face-down cards with suits $(\henc_m(uy_1),\ldots,\henc_m(uy_n))$, where $u\sample U$.

    \item[Procedures.]\mbox{}
    \begin{enumerate}
        \item Rearrange the input sequence and obtain $m$ piles $(p(0),\ldots,p({m-1}))$, where
        \begin{align*}
            p(\gamma)=\underset{c_{1,\gamma}}{\back}\cdots \underset{c_{n,\gamma}}{\back}.
        \end{align*}
        \item Rearrange the sequence of $(p(\gamma))_{\gamma\in\mathbb{Z}_m^*}$ and obtain $k$ piles $(q(0),\ldots,q({k-1}))$, where
        \begin{align*}
            q(\ell)=\underbrace{{\back}\,\cdots\,{\back}}_{p({\beta_1\alpha_\ell})}\cdots\underbrace{{\back}\,\cdots\,{\back}}_{p({\beta_t\alpha_\ell})}.
        \end{align*}
        \item Apply a pile-shifting shuffle as follows:
        \begin{align*}
            \Big\langle\underbrace{\back\cdots\back}_{q(0)}
            \Big|\cdots\Big|
            \underbrace{\back\cdots\back}_{q(k-1)}\Big\rangle
            ~\ra~
            \underbrace{\back\cdots\back}_{q(\Delta)}
            \cdots
            \underbrace{\back\cdots\back}_{q(k-1+\Delta)}\,,
        \end{align*}
        where $\Delta\sample\mathbb{Z}_k$ and indices are taken modulo $k$.
        \item Set $q'(\ell):=q(\ell+\Delta)$ for each $\ell\in\mathbb{Z}_k$ and parse it as
        \begin{align*}
            q'(\ell)=\underbrace{{\back}\,\cdots\,{\back}}_{p({\gamma_{1,\ell}})}\cdots\underbrace{{\back}\,\cdots\,{\back}}_{p({\gamma_{t,\ell}})}.
        \end{align*}
        \item Rearrange the sequence of $p(0)$ and $(q'(\ell))_{\ell\in\mathbb{Z}_k}$ and obtain $m$ piles $(p'(\gamma))_{\gamma\in\mathbb{Z}_m}$ as follows: Let $p'(0)=p(0)$. For each $\gamma\in\mathbb{Z}_m^*$, express $\gamma=\beta_s\alpha_\ell$ for some $s\in[t]$ and $\ell\in\mathbb{Z}_k$. Set $p'(\gamma)=p(\gamma_{s,\ell})$, which is the $s$-th pile of $q'(\ell)$.
        \item Let $(c_{i,\gamma}')_{i\in[n],\gamma\in\mathbb{Z}_m}$ be cards specified as follows:
        \begin{align*}
            \overbrace{\underset{c_{1,0}'}{\back}\cdots\underset{c_{n,0}'}{\back}}^{p'(0)}
            \cdots
            \overbrace{\underset{c_{1,m-1}'}{\back}\cdots\underset{c_{n,m-1}'}{\back}}^{p'(m-1)}.
        \end{align*}
        \item Output $((c_{1,\gamma}')_{\gamma\in\mathbb{Z}_m},\ldots,(c_{n,\gamma}')_{\gamma\in\mathbb{Z}_m})$.
    \end{enumerate}
\end{description}}}
\caption{A card-based protocol for multiplying inputs by $u\sample U$}
\label{figure:full-open-additive-PSM-1}
\end{figure}

\begin{proof}
Consider the protocol described in Figure~\ref{figure:full-open-additive-PSM-1}.
Let $(y_i)_{i\in[n]}\in\mathbb{Z}_m^n$.
Let $d_i=c_{i,y_i}$ denote the unique suit $\heartsuit$ contained in $\henc_m(y_i)$.
After Step~1, for each $i\in[n]$, the suit $d_i$ is moved to the $i$-th position of the pile $p(y_i)$.

For $i\in[n]$ with $y_i=0$, $d_i$ remains in the pile $p(0)$ after Step~3. 
Hence, $d_i$ is placed at the $i$-th position of $p'(0)$, i.e., $c_{i,0}'=d_i=\heartsuit$ and $c_{i,\gamma}'=\clubsuit$ for all $\gamma\neq0$.

For $i\in[n]$ with $y_i\neq0$, let $s_i\in[t]$ and $\ell_i\in\mathbb{Z}_k$ be such that $y_i=\beta_{s_i}\alpha_{\ell_i}$.
Note that at Step~4, $p(\gamma_{s,\ell})$ is the $s$-th pile of $q'(\ell)=q(\ell+\Delta)$ and $\gamma_{s,\ell}=\beta_s\alpha_{\ell+\Delta}=\beta_s\alpha_0^{\ell+\Delta+1}=\beta_s\alpha_\ell \alpha_0^\Delta$.
Thus, $\gamma_{s,\ell}=y_i$ if and only if $s=s_i$ and $\alpha_{\ell}=\alpha_{\ell_i}\alpha_0^{-\Delta}$, i.e., $\ell=\ell_i-\Delta \bmod k$.
This means that the suit $d_i$ is moved to the $s_i$-th pile $p(\gamma_{s_i,\ell_i-\Delta})$ of $q'(\ell_i-\Delta)$.
At Step~5, $d_i$ is contained in $p'(\gamma)$ with $\gamma=\gamma_{s_i,\ell_i-\Delta}=uy_i$, where $u=\alpha_0^{-\Delta}$.
Hence, $c_{i,uy_i}'=d_i=\heartsuit$ and $c_{i,\gamma}'=\clubsuit$ for any $\gamma\neq uy_i$.

In summary, $(c_{i,\gamma}')_{\gamma\in\mathbb{Z}_m}$ is identical to $\henc_m(uy_i)$ for all $i\in[n]$.
Since $\Delta$ is uniformly distributed over $\mathbb{Z}_k$, $u=\alpha_0^{-\Delta}$ is uniformly distributed over $U$.
\end{proof}

The following proposition shows a card-based protocol for adding random additive shares $(r_i)_{i\in[n]}$ of zero to $(y_i)_{i\in[n]}$.

\begin{proposition}\label{proposition:full-open-additive-PSM-2}
There exists a single-shuffle protocol such that for any $(y_i)_{i\in[n]}\in\mathbb{Z}_m^n$, a sequence of face-down cards with suits $(\henc_m(y_1),\ldots,\henc_m(y_n))$ is transformed into a sequence with suits $(\henc_m(y_1-r_1),\ldots,\henc_m(y_n-r_n))$, where $(r_i)_{i\in[n]}$ are random additive shares of zero.
The procedures involve $n-1$ pile-shifting shuffles.
\end{proposition}

\begin{figure}
\centering
\fbox{\parbox{\textwidth}{
\begin{description}
    \item[Input.] $mn$ face-down cards with suits $(\henc_m(y_1),\ldots,\henc_m(y_n))$, where $y_i\in\mathbb{Z}_m$ and $\henc_m(y_i)=(c_{i,0},\ldots,c_{i,m-1})\in\{\heartsuit,\clubsuit\}^m$.
    
    \item[Output.] $mn$ face-down cards with suits $(\henc_m(y_1-r_1),\ldots,\henc_m(y_n-r_n))$, where $(r_i)_{i\in[n]}$ are random additive shares of zero.

    \item[Procedures.] Let $X_n$ be a variable on $\mathbb{Z}_m$. Initially, set $X_n \leftarrow y_n$. For each $i = 1, 2, \ldots, n-1$, do the following procedures. 
        \begin{enumerate}
            \item Reverse the order of the integer encoding $\henc_{m}(y_i)$ except the $0$-th card as follows:
            \[
            \underbrace{\overset{0}{\back}\,\overset{1}{\back}\,\overset{2}{\back}\,\cdots\,\overset{m-1}{\back}}_{\henc_{m}(y_i)}\, \rightarrow \, 
            \underbrace{\overset{0}{\back}\,\overset{m-1}{\back}\,\cdots\,\overset{2}{\back}\,\overset{1}{\back}}_{\henc_{m}(-y_i)}.
            \]
            \item Place $\henc_{m}(X_n)$ on the bottom as follows:
            \begin{align*}
            & \overbrace{\back \, \back \, \cdots \, \back}^{\henc_{m}(-y_i)}\\
            & \underbrace{\back \, \back \, \cdots \, \back}_{\henc_{m}(X_n)}
            \end{align*}
            \item Apply a pile-shifting shuffle as follows:
            \[
            \left\langle
            \begin{tabular}{c|c|c|c}
            \back & \back & $\cdots$ & \back \rule[0mm]{0mm}{4mm}\\
            \back & \back & $\cdots$ & \back \rule[0mm]{0mm}{4mm}
            \end{tabular}\right\rangle
            \]
            Here, the upper sequence is $\henc_{m}(-y_i + r_i)$ and the bottom sequence is $\henc_{m}(X_n + r_i)$ for $r_i \sample\mathbb{Z}_m$. 
            \item Reverse the order of the upper sequence and obtain $\henc_{m}(y_i - r_i)$. Then we have $\henc_{m}(y_i - r_i)$ and $\henc_{m}(X_n + r_i)$. Update $X_n \leftarrow X_n + r_i$ and proceed to the next iteration. 
        \end{enumerate}
\end{description}}}
\caption{A card-based protocol for adding random additive shares of zeros to inputs}
\label{figure:full-open-additive-PSM-2}
\end{figure}

\begin{proof}

Consider a protocol described in Figure~\ref{figure:full-open-additive-PSM-2}.
Let $i=1$.
After the first iteration, we obtain two piles with suits $\henc_m(y_1-r_1)$ and $\henc_m(X_n+r_1)=\henc_m(y_n+r_1)$, where $r_1\sample\mathbb{Z}_m$ is a random element resulting from a pile-shifting shuffle at Step~3.

Let $i\geq 1$.
Assume that after the $i$-th iteration, we obtain $i+1$ piles with suits $\henc_m(y_1-r_1),\ldots,\henc_m(y_{i}-r_{i}),\henc_m(y_n+R_i)$, where $r_1,\ldots,r_{i}\sample\mathbb{Z}_m$ are independent random elements and $R_i=\sum_{j=1}^{i}r_j$.
Then, in the $(i+1)$-th iteration, we have $X_n=y_n+R_i$.
After Steps~3~and~4, we obtain two piles with suits $\henc_m(y_{i+1}-r_{i+1})$ and $\henc_m(x_n+r_{i+1})=\henc_m(y_n+R_i+r_{i+1})$.
We thus obtain $i+2$ piles with suits $\henc_m(y_1-r_1),\ldots,\henc_m(y_{i}-r_{i}),\henc_m(y_{i+1}-r_{i+1}),\henc_m(y_n+R_{i+1})$, where $R_{i+1}=\sum_{j=1}^{i+1}r_j$.

By induction on $i$, we finally obtain $n$ piles with suits $\henc_m(y_1-r_1),\ldots,\henc_m(y_{n-1}-r_{n-1}),\henc_m(y_n-r_n)$, where $r_1,\ldots,r_{n-1}\sample\mathbb{Z}_m$ are independent random elements and $r_n=-R_n=-\sum_{j=1}^{n-1}r_j$.
\end{proof}

Finally, by combining the above protocols, we construct a single-shuffle full-open protocol for $f$ from an additive PSM protocol for $f$.

\begin{theorem}\label{theorem:full-open-additive-PSM}
    Let $\Sigma=(\mathsf{Gen},\mathsf{Enc},\mathsf{Dec})$ be an additive PSM protocol over $\mathbb{Z}_m$ for a function $f:\bin^n\ra\bin$.
    Then, there exists a single-shuffle full-open protocol $\Pi$ for $f$ such that
    \begin{itemize}
        \item The number of cards is $mn$.
        \item The shuffle is a composition of $n$ pile-shifting shuffles.
    \end{itemize}
\end{theorem}

\begin{figure}
\centering
\fbox{\parbox{\textwidth}{
\begin{description}
    \item[Notations.]\mbox{}
        \begin{itemize}       
            \item Let $f:\bin^n\ra\bin$ be a function.
            \item Let $\Sigma=(\mathsf{Gen},\mathsf{Enc},\mathsf{Dec})$ be an additive PSM protocol over $\mathbb{Z}_m$ for $f$.
            \item Let $U$ be a subgroup of $\mathbb{Z}_m^*$ specified by $\Sigma$.
            \item Let $(g_i:\bin\ra\mathbb{Z}_m)_{i\in[n]}$ be $n$ functions specified by $\Sigma$.
        \end{itemize}

    \item[Construction of $L_i^b$.] For each $i\in[n]$ and each $b\in\bin$, let $L_i^b=\henc_m(g_i(b))$.

    \item[Shuffling.] Given a sequence of $mn$ face-down cards $\underbrace{{\back}\,\cdots\,{\back}}_{\henc_m(g_1(x_1))}\cdots\underbrace{{\back}\,\cdots\,{\back}}_{\henc_m(g_n(x_n))}\,$,
    
    \begin{enumerate}
        \item Apply the procedures in Proposition~\ref{proposition:full-open-additive-PSM-1} to obtain  
        \[
        \underbrace{{\back}\,\cdots\,{\back}}_{\henc_m(u\cdot g_1(x_1))}\cdots\underbrace{{\back}\,\cdots\,{\back}}_{\henc_m(u\cdot g_n(x_n))}\,,
        \]
        where $u\sample U$.
        \item Apply the procedures in Proposition~\ref{proposition:full-open-additive-PSM-2} to obtain
        \begin{align*}
            \underbrace{{\back}\,\cdots\,{\back}}_{\henc_m(u\cdot g_1(x_1)-r_1)}\cdots\underbrace{{\back}\,\cdots\,{\back}}_{\henc_m(u\cdot g_n(x_n)-r_n)}\,,
        \end{align*}
        where $(r_i)_{i\in[n]}$ are random additive shares of zero.
    \end{enumerate}

    \item[Output.]\mbox{}
    \begin{enumerate}
        \item Reveal all face-down cards.
        \item For each $i\in[n]$, let $\mu_i$ be the position of $\heartsuit$ in the sequence $\henc_m(u\cdot g_i(x_i)-r_i)$.
        \item Run $\mathsf{Dec}$ on input $(\mu_1,\ldots,\mu_n)$ to obtain $y\in\bin$.
        \item Output $y$.
    \end{enumerate}
\end{description}}}
\caption{A single-shuffle full-open protocol based on an additive PSM protocol}
\label{figure:full-open-additive-PSM}
\end{figure}

\begin{proof}
The protocol $\Pi$ is described in Figure~\ref{figure:full-open-additive-PSM}.
At Step~2 in the output phase, we have that $\mu_i=u\cdot g_i(x_i)-r_i$ for each $i\in[n]$.
Since $u\in U$ and $\sum_{i\in[n]}r_i=0$, $\mathsf{Dec}$ correctly outputs $y=f(x_1,\ldots,x_n)$.

Furthermore, $u$ is uniformly distributed over $U$ and $(r_i)_{i\in[n]}$ are random additive shares of zero.
Thus, the distribution of $(\mu_i)_{i\in[n]}$ is identical to that of messages of the additive PSM protocol $\Sigma$.
The privacy of $\Sigma$ ensures that $(\mu_i)_{i\in[n]}$ and hence the suits of revealed cards depend only on $f(x_1,\ldots,x_n)$.
\end{proof}

It is known that there exists an additive PSM protocol over $\mathbb{Z}_m$ for any function $f:\bin^n\rightarrow\bin$ if $m=2^{O(2^n)}$ \cite{Ish13}.
Although it is possible to apply Theorem~\ref{theorem:full-open-additive-PSM} to this additive PSM protocol, the resulting single-shuffle full-open protocol requires $mn=2^{O(2^n)}$ cards, which does not improve the number of cards $2^{O(n^32^{n/2})}$ of the construction in Corollary~\ref{corollary:SSFO-all}.

Nevertheless, Theorem~\ref{theorem:full-open-additive-PSM} improves the number of cards for special functions if we use additive PSM protocols in Propositions~\ref{proposition:psm-and}~and~\ref{proposition:psm-sym}.
The following corollaries show single-shuffle full-open protocols for the $n$-input AND function and any symmetric function, using fewer cards than Corollary~\ref{corollary:SSFO-all}.

\begin{corollary}\label{corollary:SSFO-AND}
There exists a single-shuffle full-open protocol $\Pi$ for $\mathsf{AND}_n$ such that the number of cards is $O(n^2)$ and the shuffle is a composition of $n$ pile-shifting shuffles.
\end{corollary}

\begin{corollary}\label{corollary:SSFO-Sym}
For any symmetric function $f:\bin^n\ra\bin$, there exists a single-shuffle full-open protocol $\Pi$ for $f$ such that the number of cards is $(1+o(1))n^3 2^{2n-2}$ and the shuffle is a composition of $n$ pile-shifting shuffles.
\end{corollary}

\subsection{A Single-Shuffle Full-Open Protocol for Any Function}

We extend the single-shuffle full-open protocol for $\mathsf{AND}_n$ to a protocol that applies to any function $f:\bin^n\ra\bin$.
The protocol is obtained by concatenating many protocols for $\mathsf{AND}_n$ each corresponding to an input $\bm{a}$ with $f(\bm{a})=1$.
This method requires a composition of $|f^{-1}(1)|=O(2^n)$ shuffles but achieves a smaller number of cards compared to the general solution in Corollary~\ref{corollary:SSFO-all}.

\begin{figure}
\centering
\fbox{\parbox{\textwidth}{
\begin{description}
    \item[Notations.]\mbox{}
        \begin{itemize}       
            \item Let $f:\bin^n\ra\bin$ be a function.
            \item Let $f^{-1}(1)=\{\bm{a}\in\bin^n:f(\bm{a})=1\}=\{\bm{a}^{(1)},\ldots,\bm{a}^{(N)}\}$ and denote $\bm{a}^{(j)}=(a_1^{(j)},\ldots,a_{n}^{(j)})$.
            \item Let $\Pi_0$ be a single-shuffle full-open protocol for $\mathsf{AND}_n$.
            \item Let $(\lambda_i^b)_{i\in[n],b\in\bin}$ be $2n$ sequences of suits, $\mathcal{D}_0$ be a probability distribution, $\mathsf{Output}_0$ be an output function specified by $\Pi_0$.
        \end{itemize}

    \item[Construction of $L_i^b$.] For each $i\in[n]$ and each $b\in\bin$, let $L_i^b=\lambda_i^{\phi^{(1)}(b)}||\cdots||\lambda_i^{\phi^{(N)}(b)}$, where $||$ denotes the concatenation of suits and we set $\phi^{(j)}(b)=1$ if $a_i^{(j)}=b$ and $\phi^{(j)}(b)=0$ otherwise.

    \item[Shuffling.] Given a sequence of face-down cards
    \begin{align*}
        \underbrace{\back\cdots\back}_{\lambda_1^{\phi^{(1)}(x_1)}}\cdots\underbrace{\back\cdots\back}_{\lambda_1^{\phi^{(N)}(x_1)}}
       \cdots
        \underbrace{\back\cdots\back}_{\lambda_n^{\phi^{(1)}(x_n)}}\cdots\underbrace{\back\cdots\back}_{\lambda_n^{\phi^{(N)}(x_n)}}\,,
    \end{align*}
    \begin{enumerate}
        \item For each $j\in[N]$, sample a permutation $\pi^{(j)}\leftarrow\mathcal{D}_0$ independently and apply $\pi^{(j)}$ to part of the sequence
        \begin{align*}
            \underbrace{\back\cdots\back}_{\lambda_1^{\phi^{(j)}(x_1)}}\cdots\underbrace{\back\cdots\back}_{\lambda_n^{\phi^{(j)}(x_n)}}
            ~\ra~
            \underbrace{\back\cdots\back}_{\mu^{(j)}}.            
        \end{align*}
        \item Apply a pile-scrambling shuffle as follows:
        \begin{align*}
            \Big[\underbrace{\back\cdots\back}_{\mu^{(1)}}
            \Big|\cdots\Big|
            \underbrace{\back\cdots\back}_{\mu^{(N)}}\Big]
            ~\ra~
            \underbrace{\back\cdots\back}_{\mu^{(j_1)}}
           \cdots
            \underbrace{\back\cdots\back}_{\mu^{(j_N)}}.
        \end{align*}
    \end{enumerate}

    \item[Output.]\mbox{}
    \begin{enumerate}
        \item Reveal all face-down cards.
        \item Compute $(y^{(1)},\ldots,y^{(N)})=(\mathsf{Output}_0(\mu^{(j_1)}),\ldots,\mathsf{Output}_0(\mu^{(j_N)}))$.
        \item If there exists a unique $h\in[N]$ such that $y^{(h)}=1$, then output $1$. Otherwise, output $0$.
    \end{enumerate}
\end{description}}}
\caption{A single-shuffle full-open protocol based on a single-shuffle full-open protocol for $\mathsf{AND}_n$}
\label{figure:full-open-AND}
\end{figure}

\begin{theorem}\label{theorem:full-open-AND}
    Let $\Pi_0$ be a single-shuffle full-open protocol for the $n$-input AND function.
    Then, there exists a single-shuffle full-open protocol $\Pi$ for a function $f:\bin^n\ra\bin$ such that
    \begin{itemize}
        \item The number of cards is $N\ell$, where $N=|f^{-1}(1)|=|\{\bm{a}\in\bin^n:f(\bm{a})=1\}|$ and $\ell$ is the number of cards of $\Pi_0$.
        \item The shuffle is a composition of $N$ independent shuffles for $\Pi_0$ and a pile-scramble shuffle.
    \end{itemize}
\end{theorem}

\begin{proof}
Let $\Pi$ be a card-based protocol described in Figure~\ref{figure:full-open-AND}.

Let $\bm{x}=(x_i)_{i\in[n]}$ be inputs.
The sequence of suits $(\lambda_i^{\phi^{(j)}(x_i)})_{i\in[n]}$ corresponds to an input sequence when executing $\Pi_0$ on input $(\phi^{(j)}(x_i))_{i\in[n]}$.
Thus, $\mathsf{Output}_0(\mu^{(j)})=1$ if and only if $\phi^{(j)}(x_1)=\cdots=\phi^{(j)}(x_n)=1$, i.e., $\bm{x}=\bm{a}^{(j)}$.
On the other hand, if $f(\bm{x})=1$, then there exists a unique $j\in[N]$ such that $\bm{x}=\bm{a}^{(j)}$, and if $f(\bm{x})=0$, then $\bm{x}\neq\bm{a}^{(j)}$ for all $j\in[N]$.
Therefore, if $f(\bm{x})=1$, then there exists a unique $y^{(h)}$ such that $y^{(h)}=1$ at Step~3 in the output phase, and the protocol outputs $1$. (This indicates that $\bm{x}=\bm{a}^{(j_h)}$.)
If $f(\bm{x})=0$, then $y^{(1)}=\cdots=y^{(N)}=0$ and the protocol outputs $0$.
This shows the correctness of $\Pi$.

To see privacy, let $\bm{x}=(x_i)_{i\in[n]}$ and $\tilde{\bm{x}}=(\tilde{x}_i)_{i\in[n]}$ be inputs with $f(\bm{x})=f(\tilde{\bm{x}})$.
Let $(\mu^{(j_1)},\ldots,\mu^{(j_N)})$ and $(\tilde{\mu}^{(j_1)},\ldots,\tilde{\mu}^{(j_N)})$ denote the suits of cards revealed in the output phase when executing $\Pi$ on input $\bm{x}$ and on input $\tilde{\bm{x}}$, respectively.
First, assume that $f(\bm{x})=f(\tilde{\bm{x}})=0$, i.e., $\bm{x},\tilde{\bm{x}}\notin\{\bm{a}^{(1)},\ldots,\bm{a}^{(N)}\}$.
Then, for any $j\in[N]$, both the outputs of $\Pi_0$ on input $(\phi^{(j)}(x_i))_{i\in[n]}$ and on input $(\phi^{(j)}(\tilde{x}_i))_{i\in[n]}$ are $0$.
Thus, the privacy of $\Pi_0$ ensures that the distributions of $\mu^{(j)}$ and $\tilde{\mu}^{(j)}$ are identical, which then implies that the distributions of $(\mu^{(j_h)})_{h\in[N]}$ and $(\tilde{\mu}^{(j_h)})_{h\in[N]}$ are identical.
Next, assume that $f(\bm{x})=f(\tilde{\bm{x}})=1$.
Then, during the execution of $\Pi$ on input $\bm{x}$, there exists a unique $j\in[N]$ such that $\Pi_0$ outputs $1$ on input $(\phi^{(j)}(x_i))_{i\in[n]}$.
Also, during the execution of $\Pi$ on input $\tilde{\bm{x}}$, there exists a unique $j'\in[N]$ such that $\Pi_0$ outputs $1$ on input $(\phi^{(j')}(\tilde{x}_i))_{i\in[n]}$.
Thus, the privacy of $\Pi_0$ ensures that the distribution of $\mu^{(j)}$ is identical to that of $\tilde{\mu}^{(j')}$; and that all other sequences $\mu^{(k)}$ for $k\in[N]\setminus\{j\}$ and $\tilde{\mu}^{(k')}$ for $k'\in[N]\setminus\{j'\}$ follow the same distribution (namely, the distribution of cards revealed during the execution of $\Pi_0$ with output $0$).
Since the $N$ piles $(\mu^{(1)},\ldots,\mu^{(N)})$ and $(\tilde{\mu}^{(1)},\ldots,\tilde{\mu}^{(N)})$ are randomly permuted at Step~2 in the shuffling phase, we conclude that their resulting distributions are identical.

For each $j\in[N]$, the number of cards to encode $(\lambda_i^{\phi^{(j)}(x_i)})_{i\in[n]}$ is $\ell$.
Thus, the total number of cards is $N\ell$.
\end{proof}

By applying Theorem~\ref{theorem:full-open-AND} to the protocol for $\mathsf{AND}_n$ in Corollary~\ref{corollary:SSFO-AND}, we obtain a single-shuffle full-open protocol for any function using fewer cards than Corollary~\ref{corollary:SSFO-all}.

\begin{corollary}\label{corollary:SSFO-all-2}
For any $f:\bin^n\ra\bin$, there exists a single-shuffle full-open protocol for $f$ such that the number of cards is $O(n^2 2^n)$ and the shuffle is a composition of at most $n2^n$ pile-shifting shuffles and a pile-scramble shuffle.
\end{corollary}

\section{Relation with PSM}
\label{section:relation}

Shinagawa and Nuida \cite{SN25} showed that a single-shuffle full-open protocol for $f:\bin^n\ra\bin$ using $\ell$ cards implies a PSM protocol for $f$ with communication complexity $c=\ell n$ and randomness complexity $r=\ell \log \ell +\ell n$.
Conversely, Theorem~\ref{theorem:full-open-1} shows that a PSM protocol for $f$ with communication complexity $c$ and randomness complexity $r$ implies a single-shuffle full-open protocol for $f$ using $\ell=c2^{r+1}$ cards.

Combining both results, we can relate the minimum number of cards required by single-shuffle full-open protocols for $f$ to the minimum communication/randomness complexity of PSM protocols for $f$.
We prepare notations.
For $f:\bin^n\ra\bin$, we define
\begin{itemize}
    \item $\card(f)$ as the minimum number of cards required by any single-shuffle full-open protocol for $f$.
    \item $\rand(f)$ as the minimum randomness complexity of any PSM protocol for $f$.
    \item $\comm(f)$ as the minimum communication complexity of any PSM protocol for $f$.
\end{itemize}

First, we show two lemmas regarding the communication and randomness complexity of PSM protocols.
These results are folklore to the best of our knowledge. We include formal proofs for the sake of completeness in Appendices~\ref{appendix:reduction-communication}~and~\ref{appendix:reduction-randomness}.

\begin{lemma}\label{lemma:reduction-communication}
    Let $\Sigma=(\mathsf{Gen},\mathsf{Enc},\mathsf{Dec})$ be a PSM protocol for a function $f:\bin^n\ra\bin$ with communication complexity $c$ and randomness complexity $r$.
    Then, there exists a PSM protocol $\Sigma'=(\mathsf{Gen}',\mathsf{Enc}',\mathsf{Dec}')$ for $f$ with communication complexity $c'=\min\{c,(r+1)n\}$ and randomness complexity $r'=r$.
\end{lemma}

\begin{lemma}\label{lemma:reduction-randomness}
    Let $\Sigma=(\mathsf{Gen},\mathsf{Enc},\mathsf{Dec})$ be a PSM protocol for a function $f:\bin^n\ra\bin$ with communication complexity $c$ and randomness complexity $r$.
    Then, there exists a PSM protocol $\Sigma'=(\mathsf{Gen}',\mathsf{Enc}',\mathsf{Dec}')$ for $f$ with communication complexity $c'=c$ and randomness complexity $r'=\min\{r,c2^n\}$.
\end{lemma}

Finally, we show relations among $\card(f)$, $\rand(f)$ and $\comm(f)$.

\begin{theorem}\label{theorem:relation}
For any $f:\bin^n\ra\bin$, it holds that
\begin{align}
    \frac{\rand(f)}{\log(\rand(f))+n}&\leq \card(f)\leq 2^{\rand(f)+1}(\rand(f)+1)n,\label{equation:relation-1}\\
    \frac{\comm(f)}{n}&\leq\card(f)\leq 2^{\comm(f)2^n+1}\comm(f).\label{equation:relation}
\end{align}
\end{theorem}

\begin{proof}
Let $\ell=\card(f)$, $r=\rand(f)$, and $c=\comm(f)$.
It follows from \cite{SN25} that $r\leq \ell\log\ell+\ell n=\ell\log(\ell2^n)$, i.e., $r2^n\leq (\ell2^n)\log(\ell2^n)$, and $c\leq \ell n$.
Since $a\leq b\log b$ implies $a/\log a\leq b$,\footnote{Otherwise, $b\log b<(a/\log a)\log(a/\log a)=a-a\log\log a/\log a<a$.} we have that $r/(\log r+n)\leq \ell$. 

Let $\Pi_{\mathrm{r}}$ be a PSM protocol for $f$ achieving the minimum randomness complexity $r$.
Due to Lemma~\ref{lemma:reduction-communication}, we may assume that the communication complexity of $\Pi_{\mathrm{r}}$ is at most $(r+1)n$. 
Then, it follows from Theorem~\ref{theorem:full-open-1} that $\ell\leq 2^{r+1}(r+1)n$.

Let $\Pi_{\mathrm{c}}$ be a PSM protocol for $f$ achieving the minimum communication complexity $c$.
Due to Lemma~\ref{lemma:reduction-randomness}, we may assume that the randomness complexity of $\Pi_{\mathrm{c}}$ is at most $c2^n$. 
Then, it follows from Theorem~\ref{theorem:full-open-1} that $\ell\leq 2^{c2^n+1}c$.
\end{proof}

We note that the inequalities in (\ref{equation:relation-1}) and the first one in (\ref{equation:relation}) still hold if card-based protocols and PSM protocols are restricted to the \textit{uniform} ones, by which we mean that a shuffle and randomness are sampled from uniform distributions over some sets.
Indeed, if the initial PSM protocol uses uniform randomness, then the transformation in Theorem~\ref{theorem:full-open-1} yields a card-based protocol that uses uniform shuffles only (i.e., a pile-shifting shuffle and a complete shuffle).
The transformation in \cite{SN25} produces a PSM protocol that uses uniform randomness if the initial card-based protocol uses uniform shuffles only.
The transformation in Lemma~\ref{lemma:reduction-communication} also preserves the uniformity of randomness: if the original protocol $\Sigma$ samples randomness uniformly, then so does the transformed protocol $\Sigma'$.
Therefore, the inequalities in (\ref{equation:relation-1}) and the first inequality in (\ref{equation:relation}) still hold in the uniform setting.
On the other hand, the transformation in Lemma~\ref{lemma:reduction-randomness} does not preserve uniformity since it modifies the probability mass function (see Eq.~(\ref{equation:probability}) in Appendix~\ref{appendix:reduction-randomness}).
Currently, it is not ensured that the second inequality in (\ref{equation:relation}) still holds in the uniform setting.

\bibliographystyle{abbrv}
\bibliography{card_submission}

\appendix

\section{Proof of Lemma~\ref{lemma:reduction-communication}}
\label{appendix:reduction-communication}

Let $i\in[n]$ and suppose that the bit-length $c_i$ of messages for party $i$ is larger than $r+1$.
Define a map $M_i:\bin\times\bin^r\ra\bin^{c_i}$ as $M_i(b,\rho)=\mathsf{Enc}(i,b,\rho)$ for any $b\in\bin$ and $\rho\in\bin^r$.
Define $\mathcal{M}_i=\{M_i(b,\rho)\in\bin^{c_i}:b\in\bin,\rho\in\bin^r\}$.
Let $C_i=|\mathcal{M}_i|$ and enumerate $\mathcal{M}_i=\{m_{i,1},\ldots,m_{i,C_i}\}$.
Define $\phi_i:\bin\times\bin^r\ra [C_i]$ as a map that maps $(b,\rho)$ into $k$ such that $M_i(b,\rho)=m_{i,k}$.
We define a new PSM protocol $\Sigma^{(i)}=(\mathsf{Gen}^{(i)},\mathsf{Enc}^{(i)},\mathsf{Dec}^{(i)})$ for $f$ as $\mathsf{Gen}^{(i)}=\mathsf{Gen}$,
\begin{align*}
    \mathsf{Enc}^{(i)}(j,b,\rho)
    &=\begin{cases}
        \phi_i(b,\rho)\in[C_i],~&\text{if}~j=i,\\
        \mathsf{Enc}(j,b,\rho)\in\bin^{c_j},~&\text{otherwise,}
    \end{cases}\\
    \mathsf{Dec}^{(i)}(\mu_1,\ldots,\mu_n)
    &=\mathsf{Dec}(\mu_1,\ldots,\mu_{i-1},m_{i,k},\mu_{i+1},\ldots,\mu_n),
\end{align*}
where $k=\mu_i\in[C_i]$.
Since $\lceil\log C_i\rceil\leq r+1$, the communication complexity of $\Sigma^{(i)}$ is $\sum_{j\neq i}c_j+\min\{c_i,r+1\}$ and the randomness complexity of $\Sigma^{(i)}$ is $r$.
By applying the above procedures to every $i\in[n]$ with $c_i>r+1$, we obtain a PSM protocol $\Sigma'=(\mathsf{Gen}',\mathsf{Enc}',\mathsf{Dec}')$ with communication complexity $c'=\sum_{i\in[n]}\min\{c_i,r+1\}\leq\min\{c,(r+1)n\}$ and randomness complexity $r'=r$.

\section{Proof of Lemma~\ref{lemma:reduction-randomness}}
\label{appendix:reduction-randomness}

Let $\mathcal{D}$ be a probability distribution over $\bin^{r}$ from which $\Sigma$ samples randomness.
Suppose that $|\mathrm{supp}(\mathcal{D})|>2^{c2^n}$.
For each $\rho\in\mathrm{supp}(\mathcal{D})$, define $\mathsf{Msg}_\rho:\bin^n\ra\bin^{c}$ as $\mathsf{Msg}_\rho(\bm{x})=(\mathsf{Enc}(i,x_i,\rho))_{i\in[n]}$ for any $\bm{x}=(x_i)_{i\in[n]}\in\bin^n$.
Then, there exist $\rho_1\neq\rho_2\in\mathrm{supp}(\mathcal{D})$ such that $\mathsf{Msg}_{\rho_1}$ is identical to $\mathsf{Msg}_{\rho_2}$, i.e., $\mathsf{Msg}_{\rho_1}(\bm{x})=\mathsf{Msg}_{\rho_2}(\bm{x})$ for all $\bm{x}\in\bin^n$.
Define $\mathcal{D}_1$ as a probability distribution whose support is $\mathrm{supp}(\mathcal{D}_1)=\mathrm{supp}(\mathcal{D})\setminus\{\rho_2\}$ and whose probability mass function is
\begin{align}\label{equation:probability}
    \probsub{\mathcal{D}_1}{\rho}=
    \begin{cases}
        \probsub{\mathcal{D}}{\rho},~&\text{if}~\rho\neq\rho_1,\\
        \probsub{\mathcal{D}}{\rho_1}+\probsub{\mathcal{D}}{\rho_2},~&\text{if}~\rho=\rho_1.
    \end{cases}
\end{align}
Let $\Sigma_1$ be a PSM protocol which is the same as $\Sigma$ except that randomness is chosen according to $\mathcal{D}_1$.
Since both $\rho_1$ and $\rho_2$ produce identical messages for all inputs, the distribution of messages does not change if $\mathcal{D}$ is replaced with $\mathcal{D}_1$, which means that $\Sigma_1$ is also a PSM protocol for $f$.
Note that $|\mathrm{supp}(\mathcal{D}_1)|=|\mathrm{supp}(\mathcal{D})|-1$.
We can apply the above procedures as long as the size of the support of the probability distribution of randomness is larger than $2^{c2^n}$.
Thus, we obtain a PSM protocol $\Sigma''=(\mathsf{Gen}'',\mathsf{Enc}'',\mathsf{Dec}'')$ for $f$ such that the size of the support of the distribution $\mathcal{D}''$ of randomness is $R:=|\mathrm{supp}(\mathcal{D}'')|\leq 2^{c2^n}$.
Enumerate $\mathrm{supp}(\mathcal{D}'')=\{\rho_1,\ldots,\rho_{R}\}$.
Let $r'=\lceil\log R\rceil\leq c2^n$.
Define $\mathcal{D}'$ as a probability distribution over $\bin^{r'}$ whose probability mass function is $\probsub{\mathcal{D}'}{\sigma}=\rho_\sigma$ for any $\sigma\in\bin^{r'}$ (we here identify $\sigma\in[R]$ with its binary representation).
We define $\Sigma'=(\mathsf{Gen}',\mathsf{Enc}',\mathsf{Dec}')$ as follows:
$\mathsf{Gen}'$ samples randomness $\sigma\in\bin^{r'}$ according to $\mathcal{D}'$, $\mathsf{Enc}'(i,b,\sigma)=\mathsf{Enc}''(i,b,\rho_\sigma)$, and $\mathsf{Dec}'=\mathsf{Dec}''$.
Then, $\Sigma'$ is a PSM protocol for $f$ with randomness complexity $r'=\min\{r,c2^n\}$.

\end{document}